\documentclass[journal,web,onecolumn,draftclsmode]{ieeecolor} 
\usepackage{generic}
\usepackage{cite}
\usepackage{amsmath,amssymb,amsfonts}
\usepackage{algorithmic}
\usepackage{graphicx}
\usepackage{algorithm,algorithmic}
\usepackage{hyperref}
\hypersetup{hidelinks=true}
\usepackage{textcomp}
\usepackage{bbm}
\usepackage{enumerate}
\usepackage{rotating}
\usepackage[usenames,dvipsnames]{xcolor}

\def\BibTeX{{\rm B\kern-.05em{\sc i\kern-.025em b}\kern-.08em
    T\kern-.1667em\lower.7ex\hbox{E}\kern-.125emX}}
\newtheorem{theorem}{Theorem}
\newtheorem{remark}{Remark}

\newtheorem{lemma}{Lemma}

\newtheorem{definition}{Definition}

\newcommand{\infoa}[1]{\M_{#1}}
\newcommand{\infoaN}{\M}
\newcommand{\infoc}[1]{\N_{#1}}
\newcommand{\infocN}{\N}
\newcommand{\A}{{\cal A}}

\newcommand{\hide}[1]{}

\renewcommand{\aa}{{\bf a}}

\newcommand{\I}{{\cal I}}
\newcommand{\R}{{\cal R}}
\newcommand{\C}{{\cal C}}
\renewcommand{\O}{{\cal O}}
\newcommand{\N}{{\cal N}}
\newcommand{\M}{{\cal M}}
\newcommand{\poa}{{\rm pPoA}(f,w,\C,\N)}
\newcommand{\poamc}{{\rm pPoA}{(f^{\rm mc},w,\C,\N)}}
\newcommand{\poagen}{{\rm pPoA}(\{f_i\},w,\infoaN)}
\newcommand{\GG}{{\cal G}_{f,w}^{(\C,\N)}}
\newcommand{\hGG}{\hat{\cal G}_{f,w}^{(\C,\N)}}
\newcommand{\hatGG}{\hGG}
\newcommand{\aopt}{a^{\rm opt}}
\newcommand{\ane}{a^{\rm NE}}

\newcommand{\gentuple}{(a_1,x_1,b_1),\dots,(a_k,x_k,b_k)}
\newcommand{\ttuple}{t_1\dots,t_k}
\newcommand{\tltuple}{\Tilde{t}_1\dots,\Tilde{t}_k}
\newcommand{\eop}{{\hfill $\blacksquare$}}
\newcommand{\n}{\Tilde{n}}

\newcommand{\indc}[1]{\mathbbm{1}_{\{#1\}}}

\newcommand{\fes}{f^{\rm es}}
\newcommand{\fmc}{f^{\rm mc}}
\newcommand{\fcd}{f^{\rm cd}}

\definecolor{mygreen}{rgb}{0.1, 0.7, 0.1}
\definecolor{mybrown}{rgb}{0.74, 0.27, 0.07}

\newif\ifshowrevisioncolors
\showrevisioncolorsfalse   

\newcommand{\revcolor}[2]{%
  \ifshowrevisioncolors
    \textcolor{#1}{#2}%
  \else
    #2%
  \fi
}

\newcommand{\revtwo}[1]{\revcolor{blue}{#1}}
\newcommand{\revten}[1]{\revcolor{mygreen}{#1}}
\newcommand{\revnine}[1]{\revcolor{red}{#1}}
\newcommand{\revother}[1]{\revcolor{orange}{#1}}
\newcommand{\revsixteen}[1]{\revcolor{cyan}{#1}}

\begin{document}
\title{Deriving the \revtwo{Pure} Price of Anarchy for Networked Resource Allocation Games}
\author{Vartika Singh \IEEEmembership{Member, IEEE}, 
Philip N. Brown, \IEEEmembership{Member, IEEE}
\thanks{This material is based upon work supported by the Air Force Office of Scientific Research under award number FA9550-23-1-0171, the Army Research Office under grant number W911NF-25-1-0239, and the National Science Foundation under grant number ECCS-2440836.}
\thanks{Vartika Singh conducted this research while with the University of Colorado Colorado Springs, Colorado Springs, CO, 80918 USA (email: {\tt singh.vsvartika@gmail.com}).
}
\thanks{Philip N. Brown is with the University of Colorado Colorado Springs, Colorado Springs, CO, 80918 USA (email: {\tt pbrown2@uccs.edu}). 
He performed this research while a Visiting Professor at Politecnico di Torino, Italy.
}}

\maketitle

\begin{abstract}
This work considers multi-agent coordination with arbitrary information networks among the agents using a game-theoretic approach. A system designer aims to assign local utility functions to the agents to guide their actions toward a desired system objective. The performance of the assigned local utilities is measured by the well known \revtwo{pure price of anarchy (pPoA)} metric that equals the ratio of the system objective at the worst \revtwo{pure Nash equilibrium} of the corresponding game to the optimal system objective. 
Our aim is to derive the utility functions which optimize the \revtwo{pPoA-based} performance guarantees for any given information network and system objective. %
We develop a linear program that derives \revtwo{the optimal pPoA} for any arbitrary information network and arbitrary system objective. Our work is the first to solve optimal utility design for arbitrary networks; our techniques generalize previous approaches which considered only the full-information setting.
\revtwo{For supermodular objective functions, we prove that counterintuitively, a fully communication-denied utility design is optimal irrespective of the original information network.}
For submodular system objectives, an exhaustive numerical analysis suggests that the optimal utility design is robust to communication failures even for this case. 
When the system objective is weighted maximum coverage, the marginal contribution utility design \revtwo{provably optimizes the pPoA} for a wide variety of information networks of interest.
\end{abstract}


\section{Introduction}
\label{sec:introduction}
With applications in  air traffic control \cite{atc1,atc2}, robotic networks \cite{rn1,rn3}, airport security \cite{airport}, transportation networks \cite{tn2} and  medical sciences \cite{ms1,ms2} etc, multi-agent coordination has emerged as a topic of interest among researchers. Coordinating agents in a centralized manner can be expensive or even infeasible; thus, a game-theoretic approach for distributed control has become a natural solution technique \cite{Marden, Marden2, vetta, Gairing,Seaton2023,Ramaswamy}.

In the game-theoretic approach, a system designer endows every agent with a local utility function to guide their actions.
The utility functions are selected so that the induced game's Nash equilibria achieve a favorable system objective value.
The designer then assigns agents an appropriate distributed algorithm to ensure they converge to a Nash equilibrium~\cite{Collins2025}.
Then, the performance of the assigned utility functions is measured by the \emph{price of anarchy} (PoA) metric which equals the ratio of the system objective at the worst performing equilibrium to that at the optimal \cite{PoA}. The system designer aims to assign utility functions to the agents to induce the optimal price of anarchy.

The coordination problem becomes more challenging when some agents lose access to \revother{information regarding the actions of other agents.}
This could be  due to constraints on information sharing, or a communication failure, or due to jamming by an adversary. Thus, the system designer must consider information networks while designing the utilities of the agents. 

The existing literature considers such information networks among the agents and provides lower and upper bounds on the optimal PoA \cite{Grimsman2018,Grimsman2020,Grimsman2022,Josh2022,Josh2023}.  However, most of the literature is focused on specific submodular system objective functions and specific utility functions, and lacks a methodology to derive the utility functions that optimize the PoA for general system objectives and general information networks. 
On the other hand, the literature that considers general system objectives  \cite{Marden,Marden2} focuses solely on full information networks where all agents can observe the actions of all other agents; i.e., the possibility of arbitrary information networks among the agents is not considered.

We aim to address this gap and build upon the framework of \cite{Marden} to include the information networks. In~\cite{Marden}, the authors derive the optimal PoA for general system objectives by solving a linear program (LP) for a complete network.
In the present paper, \revtwo{we focus our attention on the subproblem of optimizing \emph{pure Nash equilibria}; thus, we state all of our results in terms of the PoA defined in pure strategies known as the \emph{pure price of anarchy (pPoA).}
We develop an LP that derives the optimal pPoA for any general system objective \emph{and also any arbitrary information network.} }

\revnine{Our models combine the anonymous separable resource allocation games of~\cite{Marden} with the information network parameterization of~\cite{Josh2023} to allow for substantial modeling expressiveness.
The combination of these allows for modeling quite general settings, including facility location, target assignment, sensor coverage, and network interdiction problems as described in~\cite{vetta,Gairing,Marden-Shamma}.}

Our main contributions are as follows: (i) a linear program deriving the optimal pPoA for any system objective and information network; (ii) we prove that \revtwo{when the system objective is supermodular, the fully communication-denied marginal contribution utility is optimal regardless of the underlying information network;} (iii) the optimal \revtwo{pPoA} is at most $1/2$ for set cover games if there is any communication failure; (iv) the well-known marginal contribution utility design is optimal for a set cover game with a network of blind, isolated and normal agents; (v) if the information graph among the agents is such that each pair of agent has at least one agent able to observe the other, the marginal contribution leads to \revtwo{pPoA} of $1/2$ for set cover games.

\revsixteen{This work also contributes to the bounds in the existing literature. 
For example, for a set cover game with a blind, an isolated and two normal agents, \cite{Josh2023} provides a lower bound of 1/6, \cite{Grimsman2020} provides a lower bound of 1/2, and \cite{Grimsman2022} provides an upper-bound of 1/2 for optimal price of anarchy.
This work directly proves the optimal price of anarchy equals 1/2 for such a network.
Further, to the best of our knowledge, the literature contains no bounds for the optimal price of anarchy for supermodular games with arbitrary information networks.
This work provides the closed form expression for optimal PoA for such games, completely independent of the network.}

We also consider the robustness of the \revtwo{pPoA-optimizing} utility design to unanticipated communication failures.
For supermodular system objectives, the optimal utility design for the full information case is proved to be optimal even when a communication failure leading to arbitrary information network occurs. For submodular games, we numerically show that optimal utility functions for the full-information case perform near optimally even after a communication failure happens. Further, the numerical examples also \revtwo{seem to indicate that the optimal pPoA is monotone with respect to the connectedness of the information network.}

\section{Problem Description}
Let $\R =\{r_1, \dots,r_m\}$ be a finite set of resources where every resource $r\in \R$ is associated with a value $v_r\ge 0$. Let $N = \{1,\dots,n\}$ be a finite set of agents. Each agent $i \in N$ has action set $\A_i \subseteq 2^\R$; the set of resources selected by agent $i$ is written $a_i\in\A_i$.
The welfare generated at a resource $r$ depends on $v_r$ and the number of agents selecting $r$. Let $\A:=\A_1\times\dots\times \A_n$ denote the set of action profiles of the agents.  For $\aa = (a_1,\dots,a_n) \in \A $, define $|\aa|_r$ to be the number of agents selecting resource $r$ in action profile $\aa$. Then the  total welfare generated under action profile $\aa$, $W:\A \to \mathbb{R}$ is given by
\begin{eqnarray}\label{eqn_syst_obj}
    W(\aa) &=& \sum_{r \in \cup_i a_i} v_r w(|\aa|_r),
\end{eqnarray}where \emph{basis function} $w(|\aa|_r)>0$  scales the value of a resource $r$ depending on the number of agents selecting it.  We make the following standing assumptions:
\begin{enumerate}[\textbf{A}.1]
    \item The system objective is $W(\aa)>0$ for some $\aa \in \A$.
    \item \revtwo{The basis function $w(0)=0$ and $w(j)>0$ for $j>0$.}
\end{enumerate}

The aim of the system designer is to assign local utility functions to the agents that guide their actions towards an action profile that maximizes the  system objective \eqref{eqn_syst_obj}. If assumption \textbf{A}.1 is not true for some problem, then all the action profiles lead to the system objective value of zero  irrespective of utility functions chosen by system designer, thus \textbf{A}.1 is justified. \revtwo{Prior work~\cite{Marden} has observed that if $w(j)=0$ for $j>0$, the resulting PoA is 0; hence the restriction of assumption \textbf{A}.2. \emph{Without loss of generality, we scale $w(1)=1$ throughout the paper.}}

This framework to model a resource allocation problem is considered in \cite{Marden}, which we generalize to consider general information networks among the agents in this work.

\subsection{Local Utilities and Utility Generating Mechanism}
In contrast to prior work \cite{Marden}, we consider general information networks, where the information available to the agents can be described using an information graph  $\infoaN:=(\infoa{1},\dots,\infoa{n})$, where $\infoa{i}\subseteq N$ is the set of agents whose actions agent $i$ can observe (naturally, $i\in \infoa{i}$).
Whereas in \cite{Marden}, a full information network is considered with $\infoa{i}=N$ for all $i\in N$.

Let $[p]$ represent the set $\{1,\dots,p\}$ for any positive integer $p$ and $\{0,[p]\}$ represent $\{0,1,\dots,p\}$. Since the agents can observe a subset of agents, we modify the local utilities defined in \cite{Marden} to cater to this. Let $|\aa|^{\infoa{i}}_r$ be the number of agents in  $\infoa{i}$ that select resource $r$ in action profile $\aa$. Let $f_i: [n] \to \mathbb{R}$ be a \emph{utility generating mechanism} chosen by the system designer (for agent $i$) that scales the perceived value of a resource if it is chosen  by $|\aa|^{\infoa{i}}_r$ number of agents \revtwo{with assumption $f(0)=0$.} Then,  the  local utility functions can be defined as,
\begin{eqnarray}\label{eqn_util}
     U_i(\aa) &=& \sum_{r \in  a_i} v_r f_i(|\aa|^{\infoa{i}}_r), \hspace{5mm} \mbox{ for } i \in N.
\end{eqnarray}In \cite{Marden}, the system designer always chooses the same utility generating mechanism $f$ for all the agents, since all the agents have access to the same information i.e., $f_i=f$ for all $i\in N$. 
However, in our work, the system designer can choose different utility generating mechanisms for different agents in~\eqref{eqn_util}.

\revtwo{Three well-known examples of local utility functions and corresponding utility generating mechanisms are as follows~\cite{Marden2,Grimsman2020}.}

\subsubsection{Marginal Contribution Utility} In this mechanism, the utility of an agent equals the marginal gain in the system objective when that player participates. That is, the utility of agent $i$  at action profile $\aa$ equals \begin{eqnarray*}
    U_i(\aa) &=& \sum_{r \in  a_i} v_r [w(|\aa|^{\infoa{i}}_r)-w(|\aa|^{\infoa{i}}_r-1)].
\end{eqnarray*}The above utility is obtained by utility generating mechanism  $f^{mc}_i$ defined as
\begin{equation}\label{eqn_def_MC}
  \fmc_i(x) = w(x)-w(x-1) \mbox{ for } x \in [|\infoa{i}|] \mbox{ and } \fmc_i(0)=0.
\end{equation}

\subsubsection{Equal Share Utility}
The well-known  equal share utility for an agent $i$  in action profile $\aa$ equals 
\begin{eqnarray*}
    U_i(\aa) &=& \sum_{r \in  a_i} v_r \frac{w(|\aa|^{\infoa{i}}_r)}{|\aa|^{\infoa{i}}_r}.
\end{eqnarray*}Such a utility is obtained by utility generating mechanism  $\fes_i$ defined as
\begin{equation}\label{eqn_def_ES}
  \fes_i(x)= \frac{w(x)}{x} \mbox{ for } x\in [|\infoa{i}|]\mbox{ and } \fes_i(0)=0.
\end{equation}

\subsubsection{Communication-Denied Utility}
The communication-denied (called ``blind'' in~\cite{Grimsman2020}) utility for an agent $i$ in action profile $\aa$ equals 
\begin{eqnarray*}
    U_i(\aa) &=& \sum_{r \in  a_i} v_r .
\end{eqnarray*}Such a utility is obtained by utility generating mechanism  $\fcd_i$ defined as
\begin{equation}\label{eqn_def_CD}
  \fcd_i(x)= 1 \mbox{ for } x\geq1\mbox{ and } \fcd_i(0)=0.
\end{equation}

\subsection{Resulting Class of Games and Solution}

Once the system designer chooses a utility generating mechanism $\{f_i\}_{i\in N}$, this induces a game among the agents defined by $G = \langle \R,w, N, \{\A_i\}_{i \in N},\infoaN, \{f_i\}_{i\in N} \rangle$. This work focuses on the solution concept of pure strategy Nash equilibrium \revtwo{(PNE)} for any $G$, defined as follows. For any $\aa$, let $(\Tilde{a}_i,a_{-i})$ represent the action profile where agent $i$ unilaterally deviates to $\Tilde{a}_i$ from $a_i$. Then \revtwo{PNE} is defined as
\begin{definition}[Pure Strategy Nash Equilibrium]
   For any $G$, an action profile $\aa$ is a pure strategy Nash equilibrium if $U_i(\aa) \ge U_i(\Tilde{a}_i,a_{-i})$ for all $\Tilde{a}_i \in \A_i$ and all  $i \in N$.
\end{definition}

\revtwo{The utility design literature has broadly adopted PNE as a solution concept of interest, due to the fact that in many relevant classes of games, games are guaranteed to possess a PNE and simple learning distributed dynamics are guaranteed to converge to a PNE~\cite{learning,Marden-Shamma,Marden}.
Thus, establishing bounds on the performance of pure Nash equilibria is of interest to the system designer since such a bound serves as an algorithmic approximation ratio for an entire class of distributed algorithms.}

\revtwo{In the context of this paper, it is important to note that the existence of PNE is dependent on a combination of the utility generating mechanism and the specific information network; there are combinations of network and utility generating mechanism such that some game instances do not possess a PNE.}
\revtwo{Nonetheless, as we will discuss, many meaningful results can still be shown by focusing solely on PNE.
Thus, for the remainder of the paper {we restrict our attention only to the games that have at least one pure strategy Nash equilibrium.}
The specific ramifications of this restriction will be discussed where applicable; in some cases, we will even show that this restriction is actually made without loss of generality.}

\revtwo{Let $\aopt$ denote an optimal action profile. The aim of the system designer is to choose utility generating mechanisms $\{f_i\}_{i\in N}$ such that $W(\ane)$ is `close' to $W(\aopt)$ for all PNE, $\ane$, of the resulting game.} In general, the exact resource set $\R$ and exact action sets of the agents $\{\A_i\}_{i\in N}$ may not be known at design-time.
The system designer may only know the basis function $w(\cdot)$ and information graph $\infoaN$, and choose mechanism $\{f_i\}_{i\in N}$ based on this information. Thus it is important to derive the performance guarantee of $\{f_i\}_{i\in N}$ for all possible games.

For any $w$ and $\M$, the system designer chooses $\{f_i\}_{i\in N}$ which leads to a class of games ${\cal G}_{\{f_i\},w}^\infoaN$ that have
\begin{enumerate}[\textbf{G}.1]
    \item $n$ agents,  any possible $\R$, any possible $\{\A_i\}_{i\in N}$,
    \item system objective $W(\cdot)$ defined in \eqref{eqn_syst_obj} with $w(\cdot)$ as the basis function,
    \item local utilities as in \eqref{eqn_util} with $\{f_i\}_{i\in N}$ as  utility generating mechanism and  $\infoaN$ as the information graph. 
    \item \revtwo{at least one PNE.}
\end{enumerate}Then, the performance of $f$ can be measured using the Price of Anarchy (PoA) metric (see \cite{PoA}), defined as
\revtwo{
\begin{eqnarray}\label{eqn_def_poa}
  \poagen = \inf_{G \in {\cal G}_{\{f_i\},w}^\infoaN}\left(\frac{\min_{\aa \in {\rm PNE}(G)} W(\aa)}{\max_{\aa \in \A}W(\aa)}\right),
\end{eqnarray}where ${\rm PNE}(G)$ is the set of all PNE of $G\in {\cal G}_{\{f_i\},w}^\infoaN$\footnote{Here, for any game $G$ which possesses no PNE (i.e., ${\rm PNE}(G)=\emptyset$), we define the ratio in~\eqref{eqn_def_poa} to be 1.
However, note that ${\cal G}_{\{f_i\},w}^\infoaN$ trivially always possesses at least one game $G$ such that ${\rm PNE}(G)\neq \emptyset$, so~\eqref{eqn_def_poa} is always well-defined. }}

\revtwo{In order to derive the performance guarantee of $\{f_i\}_{i\in N}$, we will define a linear program which will also turn out to be instrumental in deriving the utility generating mechanism  $\{f^{opt}_i\}_{i\in N}$ that will result in optimal performance.} The next section presents the required structure which facilitates this.


\textbf{Remark:} This work  considers the resource allocation problem considered in \cite{Marden} and extends the problem to general information networks among the agents. The introduction of general information networks makes the resource allocation problem widely applicable in presence of possible communication failures, constrained information sharing, \revother{and other applications.}

\begin{figure}
    \centering  
    \includegraphics[trim ={4cm 7.5cm 4cm 7.5cm},clip,scale=0.45]{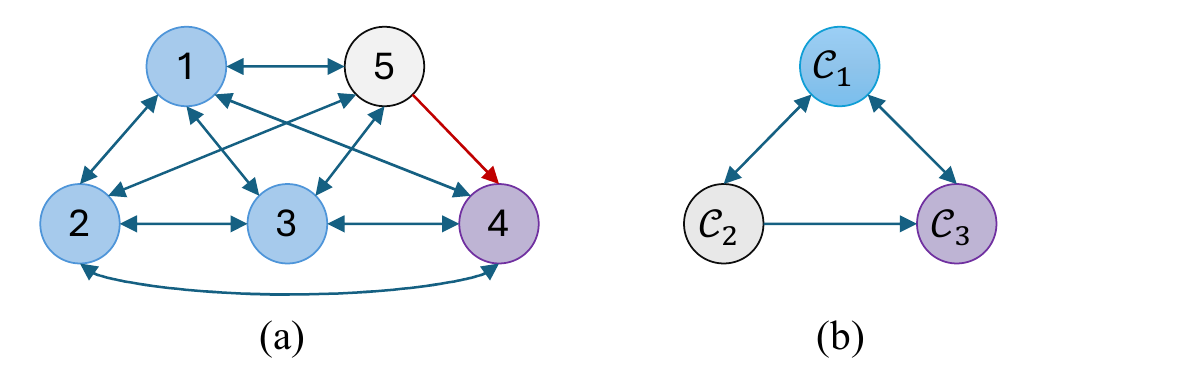}
    \caption{A general network (a) and corresponding classes (b). All the agents can see the actions of all other agents except agent $4$ who can not observe the actions of agent $5$. Thus, the information available to various agents equals $\infoa{i}  = \{1,2,3,4,5\}$ for $i=1,2,3,5$ and $\infoa{4}=\{1,2,3,4\}$. Agents $1$, $2$ and $3$ observe the same set of agents and are observed by the same set of agents, hence belong to the same class $\C_1$. Agent $5$ also observes the same set of agents but cannot be observed by agent $4$, thus belongs to class $\C_2$. Agent $4$ belongs to class $\C_3$.  Then $\infoc{1}= \infoc{2}= \{\C_1,\C_2,\C_3\}$ and $\infoc{3} = \{\C_1 , \C_4\}$, and \textbf{C}.1-\textbf{C}.3 are satisfied. 
}
    \label{fig_networks}
\end{figure}

\subsection{Information Graph and Classes}
 One can utilize the given information graph $\infoaN$ to  partition the set of agents to  classes of agents with `similar attributes'. We say two agents $i$ and $\Tilde{i}$ are similar if (i) they can observe the same set of agents, that is, $\infoa{i}=\infoa{\Tilde{i}}$, and (ii) if they can be observed by the same set of agents, that is,  $\{\infoa{i'}: i \in \infoa{i'},\ i'\in N\}=\{\infoa{i'}: \Tilde{i} \in \infoa{i'},\ i'\in N\}$. The motivation behind such a partition stems from the fact that similar agents affect the game in a similar manner; thus, the system designer can assign them the same local utility functions.  Let $\C_j$ represent a class of similar agents and $\infoc{j}$ be the set of agents whose actions agents in class $\C_j$ can observe (see Figure \ref{fig_networks}).  
%
%
A partition of $N$ such that similar agents are assigned to the same class is accomplished by  \revtwo{creating classes} where
\begin{enumerate}[\textbf{C}.1]
    \item every agent is assigned to exactly one class,
    \item every agent in class $\C_j$ observes every other agent in $\C_j$,
    \item \revtwo{every agent in  $\C_j$ observes every agent in classes $\N_j$ given by $\{\C_{m(j)}:m(j)\in \O_j\}$ where $\O_j\subseteq [k]$.}
\end{enumerate}Observe that \textbf{C}.2 along with \textbf{C}.3 implies that $\infoa{i} =\infoa{i'}$ for $i,i'\in \C_j$. Further, \textbf{C}.3 also implies that all the agents in a class are observed by the same set of agents.\footnote{\revtwo{It is important to note that the partitioning of any network according to \textbf{C}.1-\textbf{C}.3 is not restrictive. Indeed,  any network can be partitioned by  making singleton classes $\C_i=\{i\}$ for $i\in N$.}}
Then, any network can be completely specified using $(\C,\infocN)$ where $\C=\{\C_1,\dots \C_k\}$ is  the set of classes  of agents for some $k \le n$, and $\infocN=\{\infoc{1},\dots,\infoc{k}\}$ represents the information available to various classes. Thus from here onward we refer to $(\C,\N)$ as the information network. 

\revtwo{Throughout the paper, $i\in N$ represents an agent, $\C_j$ represents a class with $j\in[k]$ where $k$ represents the number of classes in the network $(\C,\N)$.}
 
\subsection{\revtwo{Pure Price of Anarchy} for Network $(\C,\infocN)$} 
For any information network  $(\C,\N)$, all the agents in any class $\C_j$ are similar hence the system designer can assign them the same utility generating mechanism. Let $\kappa_j$ represent the number of agents in class $\C_j$. Let $s_j = \sum_{l\in[k]} \kappa_l \indc{\C_l \in \N_j}$  be the total number of agents that any $i \in \C_j$ can observe (note that $s_j = |\infoa{i}|$ for all $i\in \C_j$). Then we use $f_j:[s_j] \to \mathbb{R}$ to represent the utility generating mechanism for class $\C_j$ for any $j\in [k]$, that is, for all $i \in \C_j$ the utility generating mechanism  $f_i =f_j$.

 Let $|\aa|^{\infoc{j}}_r$ be the number of agents in $\infoc{j}$ selecting resource $r$ in action profile $\aa$. Then,
the utility of an agent $i \in \C_j$ defined in \eqref{eqn_util} modifies to
\begin{eqnarray}\label{eqn_util_gen}
     U_i(a_i,a_{-i})&=& \sum_{r \in  a_i} v_r f_j(|\aa|^{\infoc{j}}_r). 
\end{eqnarray}

When the system designer chooses a set of utility generating mechanisms represented by $f=\{f_j\}_{j\in[k]}$ for a given network $(\C,\N)$ and basis function $w$, this defines a class of games ${\cal G}_{f,w}^{(\C,\infocN)}$ that satisfy $\textbf{G}.1-\textbf{G}.2$, and  have
\begin{enumerate}[\textbf{G}.1]
\setcounter{enumi}{4}
   \item local utilities given by \eqref{eqn_util_gen} with $\{f_j\}_{j\in[k]}$ as  utility generating mechanism and  $(\C,\infocN)$ as the network, where $\C =\{\C_1,\dots,\C_k\}$ and $\infocN=\{\infoc{1},\dots,\infoc{k}\}$ for $k \le n$. 
\end{enumerate}The performance of $f$  for the information network $(\C,\infocN)$ with basis function $w$ and  utility generating mechanism $f$ can be measured by modifying the \revtwo{pPoA \eqref{eqn_def_poa} to the following:\footnote{As in~\eqref{eqn_def_poa}, for any game $G$ which possesses no PNE (i.e., ${\rm PNE}(G)=\emptyset$), we define the ratio in~\eqref{eqn_poa_gen} to be 1.
However, note that $\GG$ trivially always possesses at least one game $G$ such that ${\rm PNE}(G)\neq \emptyset$, so~\eqref{eqn_poa_gen} is always well-defined. }
\begin{eqnarray}\label{eqn_poa_gen}
    \poa= \inf_{G \in \GG}\left(\frac{\min_{\aa \in {\rm PNE}(G)} W(\aa)}{\max_{\aa \in \A}W(\aa)}\right).
\end{eqnarray}}

\revtwo{The authors in \cite{Marden} provide a linear program that characterizes the pPoA} for any utility generating mechanism for the full information case; this can be viewed as a special case of our problem with information network $(\C,\N)$ where $\C=\{\C_1\}$, ${\cal N}_1=\C_1$ and $\C_1=N$. Further, they provide another linear program that derives the optimal utility generating mechanism, \revtwo{again only for the full information case, and requiring all agents to have the same utility generating mechanism.}

In this work, we aim to derive a linear program that caters to general information networks and solves \eqref{eqn_poa_gen}.

\section{Supermodular Games}
\label{sec:supermod}
We first consider the resource allocation games where the system objective is supermodular and non-decreasing\cite{Marden}. A supermodular objective function is achieved by non-decreasing and convex basis function $w(\cdot)$. That is, we assume that the basis function $w(\cdot)$ satisfies the following:
\begin{eqnarray}\label{eqn_def_sup_mod}
    w(j+1) &\ge& w(j) \mbox{ with }  w(0)=0 \mbox{ and }  \\
   w(j+1)-w(j)  &\ge& w(j)-w(j-1) \mbox{ for all } j\in [n-1] . \nonumber
\end{eqnarray}
The first main result presented in Theorem \ref{thm_sup_mod}  derives the \revtwo{pPoA}  for supermodular games with arbitrary information networks. Authors in \cite{Marden2} provide a similar result for the special case where all the agents belong to one class, $\C = \C_1$ and $\infoc{1} = \C_1$; in other words, all the agents can observe the actions of all the other agents. Theorem \ref{thm_sup_mod} (proof in Appendix \ref{appen_remaining_proofs}) extends that result to any arbitrary information network.

\revtwo{\begin{theorem}\label{thm_sup_mod}
  For any supermodular game with $n$ agents, information network $(\C,\infocN)$, and basis function $w(\cdot)$ satisfying \eqref{eqn_def_sup_mod},
    \begin{enumerate}[(i)]
        \item For any $f$, the pure price of anarchy satisfies 
        \begin{equation}\label{eqn_opt_poa_sup_mod}
            \poa \le \frac{n}{w(n)}.
        \end{equation}
        \item The communication-denied utility mechanism $\fcd:= \{\fcd_j\}_{j\in[k]}$ as in \eqref{eqn_def_CD} obtains the upper bound on PoA in \eqref{eqn_opt_poa_sup_mod}.
    \end{enumerate}
\end{theorem}}
\vspace{3mm}

\revnine{Note that the upper bound on pPoA for a resource allocation game with supermodular basis function does not depend upon the network structure.}
\revtwo{Further, the upper bound is achieved by communication-denied utility mechanism and any game is guaranteed to have a PNE under this utility mechanism. This makes communication-denied utility mechanism an \textit{optimal utility mechanism} for supermodular games. The optimal pPoA does not change when one has a full information case, or arbitrary information network where some of the agents cannot observe the actions of some other agents. }

\revtwo{Further, \emph{because it preemptively severs communication, the communication-denied utility mechanism is optimal for any arbitrary information network $(\C,\infocN)$, and hence by construction it is robust against unanticipated communication failures among the agents.}
Note that the above theorem does not claim that the optimal utility generating mechanism is unique.
Indeed,~\cite{Marden2} proves that the equal share utility generating mechanism~\eqref{eqn_def_ES} is optimal for the full-information case; but interestingly it achieves the same pPoA bound as in our Theorem~\ref{thm_sup_mod} in~\eqref{eqn_opt_poa_sup_mod}. 
Our result demonstrates the essential robustness of multiagent coordination in supermodular games.}
\revnine{Note also that if network dynamics are sufficiently slow compared to agent action dynamics, this result essentially applies to dynamic networks as well.}

\revtwo{
\begin{remark}
It is important to point out that Theorem~\ref{thm_sup_mod}'s restriction to the pure price of anarchy is without loss of generality: the upper bound in~\eqref{eqn_opt_poa_sup_mod} applies to mixed Nash equilibria as well, and the communication-denied utility generating mechanism always results in games which are guaranteed to possess at least one pure Nash equilibrium (proved in~\cite{Grimsman2020}).
\end{remark}
}

\section{Submodular Games}
\label{sec:submod}
Now we present the results on \revtwo{pure} price of anarchy for another very important class of resource allocation problems, which have a non-decreasing and submodular objective function (e.g.,\cite{Grimsman2020,Josh2022}). Such an objective function is achieved by a non-decreasing and concave basis function. That is, the basis function $w(\cdot)$ satisfies the following:
\revten{
\begin{eqnarray}\label{eqn_def_sub_mod}
    w(j+1) &\ge& w(j) \mbox{ with }  w(0)=0 \mbox{ and }  \\
    w(j+1)-w(j)  &\le& w(j)-w(j-1) \mbox{ for all } j\in [n-1] . \nonumber
\end{eqnarray}}
For submodular games, a result analogous to Theorem~\ref{thm_sup_mod} does not hold; the \revtwo{pure price of anarchy} can depend upon the information network $(\C,\N)$ \revtwo{and basis function $w$} under consideration. Thus, we provide the results for information networks that are commonly considered in the literature \cite{Grimsman2020,Josh2022}, while the case of any arbitrary information networks (and arbitrary basis function) is deferred to Section \ref{sec_arbitrary_games}.
\revnine{Nonetheless, note that our theorems in this section give analytical results which apply to broad classes of information networks, rather than very specific fixed ones.}

\subsection{Set Cover Games}
Consider the special class of submodular games widely known as \emph{set cover games} \cite{Gairing}, where the basis function $w$ is defined as
\begin{equation}\label{eqn_def_set_cov}
    w(j) =1  \mbox{ for all } j\ge1  \mbox{ and } w(0)=0.
\end{equation}

Such a basis function implies that once a resource is selected by an agent, there is no gain in system objective if more agents select the same resource.  

\subsubsection{A General Upper Bound on Networked Set Cover Games}
The next theorem  (proof in Appendix \ref{appen_remaining_proofs}) provides an upper bound on the \revtwo{pure price of anarchy} for set cover games for information networks where some agents can observe the actions of only a subset of agents. 
\begin{theorem}\label{thm_set_cov_com_fail}
     Consider a set cover game with $n$ agents, information network $(\C,\infocN)$, and basis function $w(\cdot)$ satisfying \eqref{eqn_def_set_cov}.  \revtwo{If the information network is not full-information, i.e.,} there exist classes $\C_l \in \C$ and $\C_p \in \C$ such that $\C_p \not \in \infoc{l}$, then
     \revtwo{\begin{eqnarray}
          \poa &\le& \frac{1}{2}.
     \end{eqnarray}}
\end{theorem}
\vspace{3mm}

If the hypothesis of above theorem is not satisfied by a network $(\C,\N)$, that implies that  all the agents can observe the actions chosen by all other agents leading to a full information network. For set cover games with full information, it is known that the optimal PoA is $1-1/e \approx 0.63$~\cite{Gairing}. 
However, our Theorem \ref{thm_set_cov_com_fail} implies that if there is even one agent whose actions can not be observed by some other agent, the \revother{pure} price of anarchy immediately reduces to $\frac{1}{2}$. This result implies that if there is a communication failure in a system leading to some agents not being able to observe the others, the PoA will always be less than half irrespective of the size of the network, number of agents in the system and the choice of utility generating mechanism.

\subsubsection{Set Cover Games with Blind and Isolated Agents}
An interesting case of information networks is the one with blind, isolated and normal agents. We borrow the terminology of blind and isolated agents from \cite{Grimsman2020}, and define these agents as follows:
\begin{enumerate}
    \item \textbf{Blind Agents:}  \revtwo{These are the agents who cannot observe anyone but themselves; however, they can be observed by other agents in the system. Formally, an agent $i \in \C_j$ for some $j\in[k]$   is \emph{blind} if $\C_j = \{i\}$ and $\infoc{j} = \C_j$, and there exists some $p \in [k]\backslash\{j\}$ such that $\C_j \in \infoc{p}$.} Let $\C_B\subseteq\C$ represent the collection of classes of blind agents. 
    \item \textbf{Isolated Agents:} \revtwo{These agents cannot observe any other agents but themselves. 
    Further, they cannot be observed by any other agents. Formally, an agent $i \in \C_j$ for some $j\in[k]$   is \emph{isolated} if $\C_j = \{i\}$ and $\infoc{j} = \C_j$, and  $\C_j \not\in \infoc{p}$ for any $p \in [k]\backslash\{j\}$.} 
    Let $\C_I\subseteq\C$ represent the collection of classes of isolated agents. 
    \item \textbf{Normal Agents:} \revtwo{These are the agents who can observe all the agents except isolated agents. Formally, an agent $i \in \C_j$ for some $j\in[k]$   is \emph{normal} if $\C_j = \{i: i\not\in \C_B \cup \C_I\}$, and $\infoc{j} = \{\C_j,\C_B\}$. }
\end{enumerate}

For a network consisting of blind, isolated and normal agents, we have the following result  (proof in Appendix \ref{appen_remaining_proofs}). 
\begin{theorem}\label{thm_bl_is_sc}
   Consider a set cover game with information network $(\C,\infocN)$ such that there are $\kappa_1$ blind agents, $\kappa_2$ isolated agents and $n-\kappa$ normal agents where $\kappa:=\kappa_1+\kappa_2>0$ and  with basis function $w$ satisfying \eqref{eqn_def_set_cov}. \revtwo{Then,
   \begin{enumerate}[(i)]
\item The pure price of anarchy satisfies
\begin{equation}\label{eqn_opt_poa_bl_is_sc}
           \poa \le \max\left\{\frac{1}{1+\kappa},\frac{1}{n}\right\}.
       \end{equation}
       \item If $\kappa_1>0$, the  marginal contribution utility mechanism $\fmc:= \{\fmc_j\}_{j\in[k]}$ as in \eqref{eqn_def_MC}  obtains the bound in \eqref{eqn_opt_poa_bl_is_sc}.
       \item If $\kappa_1=0$, the  marginal contribution utility mechanism $\fmc$
       as in \eqref{eqn_def_MC} satisfies
       \begin{eqnarray}\label{eqn_mc_poa_bl_is_sc}
           \poamc =  \max\left\{\frac{1}{2+\kappa},\frac{1}{n}\right\}.
       \end{eqnarray}
       \end{enumerate}  }
\end{theorem}

\vspace{3mm}

\revtwo{For a network with blind, isolated and normal agents, the existence of PNE is guaranteed for any utility generating mechanism. By above theorem, for such a network with at least one blind agent, \emph{the marginal contribution utility emerges to be optimal utility generating mechanism leading to the price of anarchy in \eqref{eqn_opt_poa_bl_is_sc}}. In fact} \cite[Theorem 2]{Grimsman2020} derives the PoA at marginal contribution for set cover games for the network with blind, isolated and normal agents, \revtwo{which matches with the PoA of \eqref{eqn_opt_poa_bl_is_sc} and  \eqref{eqn_mc_poa_bl_is_sc} in respective cases}. Thus, our present theorem completes the characterization from~\cite{Grimsman2020} by further proving that there is no other mechanism that can perform better than marginal contribution in set cover games with blind, isolated and normal agents.  \hide{However, one can derive the exact value of optimal PoA and optimal utility generating mechanism using Theorem \ref{thm_opt_lp_gen} of Section \ref{sec_arbitrary_games}.}

\subsubsection{Set Cover Games with General Networks}
We now present a result  that is applicable to a more general network, and depicts the robustness of the game-theoretic approach to the communication failure provided that a certain level of `connectedness' in the information graph $\infocN$ is maintained  (proof in Appendix \ref{appen_remaining_proofs}).

\begin{theorem}\label{thm_clique}
     Consider a set cover game with $w(\cdot)$ as in \eqref{eqn_def_set_cov} and an information network $(\C,\infocN)$ such that \revtwo{out of any two agents at least one observes the other but not all agents
can see all other agent. In other words,}  for every $\C_j \in \C$ and $\C_l \in \C$, either $\C_j \in \infoc{l}$ or $\C_l \in \infoc{j}$ or both. Further, there exist classes $\C_l \in \C$ and $\C_p \in \C$ such that $\C_p \not \in \infoc{l}$. \revtwo{Then the  marginal contribution utility mechanism $\fmc:= \{\fmc_j\}_{j\in[k]}$ obtains the PoA $\frac{1}{2}$.}
     \end{theorem}
\vspace{3mm}

From Theorem \ref{thm_set_cov_com_fail}, the \revother{pure price of anarchy} is at most half if there is even one class $\C_j$ that cannot observe the actions of agents in another class $\C_l$. However, Theorem \ref{thm_clique} implies that even if there are multiple pairs of classes $\C_j$ and $\C_l$ such that $\C_j$ cannot observe the actions of agents in $\C_l$, as long as $\C_l$ can observe the actions of agents in $\C_j$, then the \revtwo{pPoA under marginal contribution} does not degrade any further and equals $1/2$.
Figure \ref{fig_netowrk_ex}(a) presents an example of such a network with three classes.

\begin{figure}
    \centering
    \includegraphics[trim={6cm 9.7cm 6cm 6.5cm},clip,scale = 0.45]{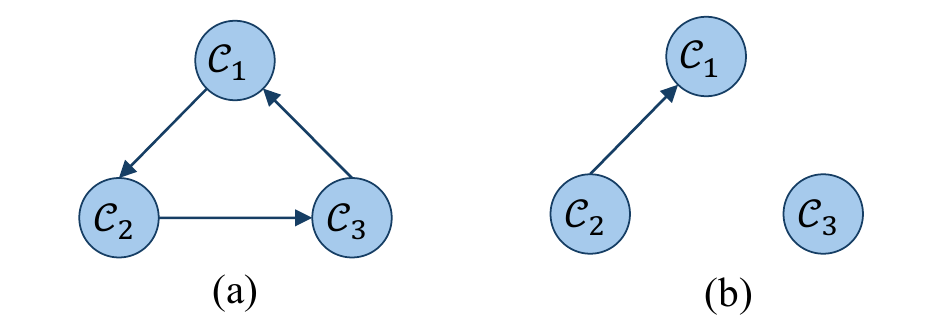}
    \caption{A network with three classes such that for every $\C_j \in \C$ and $\C_l \in \C$, either $\C_j \in \infoc{l}$ or $\C_l \in \infoc{j}$ (left). From Theorem \ref{thm_clique}, the \revtwo{PoA under marginal contribution for this network is $\frac{1}{2}$  irrespective of the number of agents in any class.} The right sub-figure has another network with three classes, such that $\N_1 = \{\C_1\}$,  $\N_2 = \{\C_1,\C_2\}$ and $\N_3 = \{\C_3\}$. \revtwo{By Theorem \ref{thm_gen_sc}, the pure PoA under marginal contribution  satisfies $\frac{1}{3} \le \poa \le \frac{1}{2}$.}}
    \label{fig_netowrk_ex}
\end{figure}

The next theorem  provides a result applicable to more general class of networks that do not \revtwo{have the same level of connectedness as in} Theorem \ref{thm_clique}. \revtwo{Consider the isolated classes and blind classes as generalization of isolated and blind agents, where a group of agents loses the ability to observe the actions of other agents outside the group. A \emph{blind class} is one that only
observes its own class but is observed by some other class, an \emph{isolated class} is the one that only observes its
own class and furthermore is not observed by any other class, and a \emph{normal class} is one that is
neither blind nor isolated. Formally, $\C_j$ is an isolated class if $\N_j=\C_j$ and $\C_j \not \in \N_l$ for $l\ne j$, and $\C_j$ is a blind class if $\N_j=\C_j$ and $\C_j \in \N_l$ for some $l\ne j$. For a network with blind and isolated classes, we have the following (proof in Appendix \ref{appen_remaining_proofs}):}

\begin{theorem}\label{thm_gen_sc}
\revtwo{Consider a set cover game with network $(\C,\N)$ with $k_i$ isolated, $k_b$ blind and $k_n$ normal classes. Define $k_1:=k_i+k_b$ and note that $k= |\C| = k_i+k_b+k_n$.  Then,
\begin{enumerate}[(i)]
  \item When $k_n<k$ the pure price of anarchy satisfies
    \begin{equation}
    \poa\le \max\left\{\frac{1}{1+k_1},\frac{1}{k}\right\}.
    \end{equation}
    \item If there is at least one blind class, i.e., $k_b,k_n\ge 1$, then then the  marginal contribution utility mechanism $\fmc:= \{\fmc_j\}_{j\in[k]}$ satisfies 
    \begin{equation}
    \poamc \ge \frac{1}{k} \label{eq:Theorem 5 part i}
    \end{equation}
    \item If there is no blind class, i.e., $k_b=0$ and $k_i,k_n \ge 1$ or no normal class, i.e., $k_n =0$, the  marginal contribution utility mechanism $\fmc:= \{\fmc_j\}_{j\in[k]}$ satisfies
    \begin{equation}\poamc \ge \frac{1}{1+k}.\label{eq:Theorem 5 part ii}
    \end{equation}
    \end{enumerate}}
\end{theorem}
\vspace{3mm}

Interestingly, in both Theorem \ref{thm_clique} and Theorem \ref{thm_gen_sc}, \revtwo{the bounds on the  PoA} do not depend upon the number of agents in various classes, and only on the number of classes in the network $(\C,\N)$. Figure \ref{fig_netowrk_ex}(b) presents an example of a network with three classes that satisfies the condition in \revother{part (ii)} of Theorem \ref{thm_gen_sc}. 

\subsection{General Submodular Games}Now we consider the case of general submodular games.  Section \ref{sec_arbitrary_games} provides a linear program (LP) that derives the \revother{pure price of anarchy  for any given basis function, utility generating mechanism and any arbitrary information network.}
This allows for numerical computation of the \revother{pure price of anarchy} even in cases when closed-form expressions are difficult to obtain.
We solve the LP and derive the \revother{pure} PoA for various submodular basis functions and present some interesting observations. In the following, submodular basis functions of the form $w(j) = j^d$ with $d \in [0,1]$ are considered; $d=0$ corresponds to set cover games, and $d>0$ covers a wide class of submodular functions. Hence, we refer to $d$ as the \emph{submodularity parameter.}

\revother{In this section, we restrict our attention to the information networks with normal, blind and isolated agents. For such games, \emph{the PNE is guaranteed to exist for any utility generating mechanism}. Hence one can optimize over utility generating mechanism $f$ in the LP in \eqref{eqn_poa_dual_lp_gen} to derive optimal utility generating mechanism and optimal price of anarchy. The non-linear terms of the form $\lambda_jf_j$ that arise upon considering $f$ as a decision variable is taken care} \revtwo{by the virtue of Lemma \ref{lem_scaling_f}, which implies that \emph{the scaling of $f_j$ with a positive coefficient doesn't alter PoA}.} \revother{See Remark \ref{remark_opt_Poa_for_bl_is} for more details. The code for computing/optimizing PoA is available at \cite{Simulation_Github}.} \revtwo{All the observations presented in the section are empirical.}

\subsubsection{Robustness} We first study the robustness of the optimal full information utility generating mechanism \revother{represented by $f^*$ from \cite[Theorem 4]{Marden}}, and show that it is quite robust to communication failures. In Figure \ref{fig_robustness_vs_bl}, a submodular basis function with a low submodularity parameter ($d=0.25$) in the left sub-figure and a high submodularity parameter ($d=0.75$) in the right sub-figure is considered. The LP in Section \ref{sec_arbitrary_games} provides the \revother{pure price of anarchy for $f^*$} and \revtwo{one can calculate optimal PoA by further optimizing over $f$ for the networks with blind/isolated agents}. The PoA when agents continue using $f^*$ even when there is a communication failure is not significantly different than the  optimal PoA specific to the information network, indicating significant robustness to communication failures. 
\begin{figure}\begin{center}
        \includegraphics[trim ={3.6cm 8cm 4cm 8cm},clip,scale = 0.5]{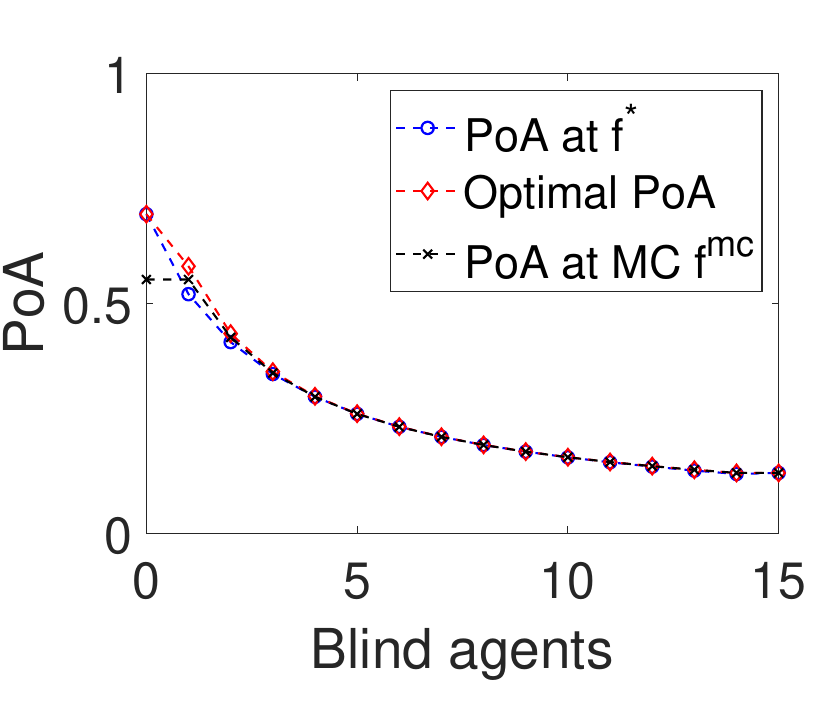}
         \includegraphics[trim ={2.9cm 8cm 4cm 8cm},clip,scale = 0.5]{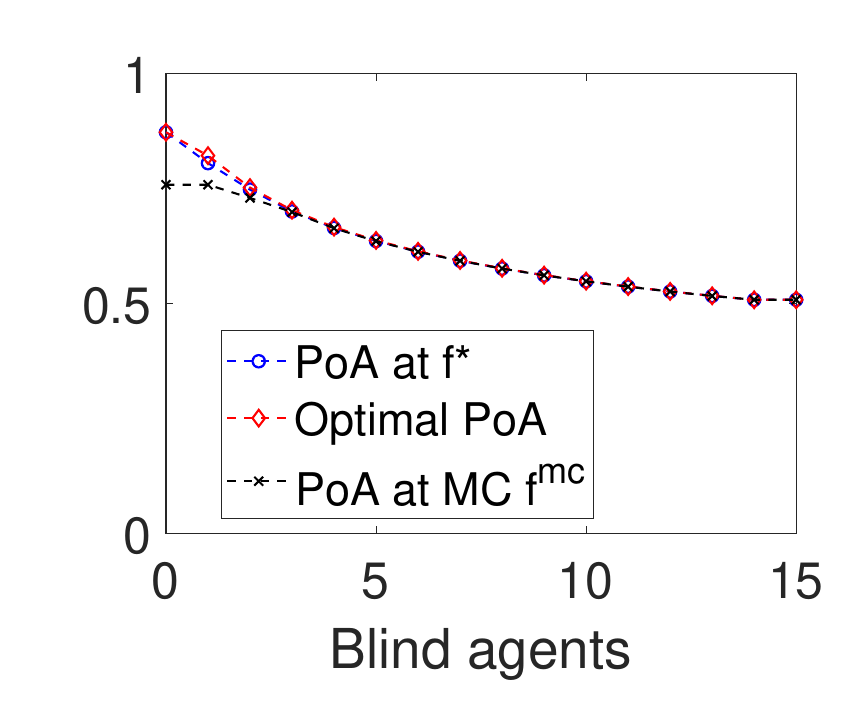}
         \end{center}
    \caption{The optimal utility generating mechanism $f^*$ corresponding to the full information case  where all the agents can observe the actions of all other agents  is `almost' optimal for any network $(\C,\N)$ with blind and normal agents. With low submodularity  parameter $d=0.25$ (left), the marginal contribution utility mechanism $\fmc$ performs better than $f^*$ but with high submodularity parameter $d=0.75$ (right) $f^*$ outperforms $\fmc$. Thus, $f^*$ is robust against communication failures.}
\end{figure}
\begin{figure}
    \label{fig_robustness_vs_bl}
       \includegraphics[trim ={3.6cm 8cm 4cm 8cm},clip,scale = 0.5]{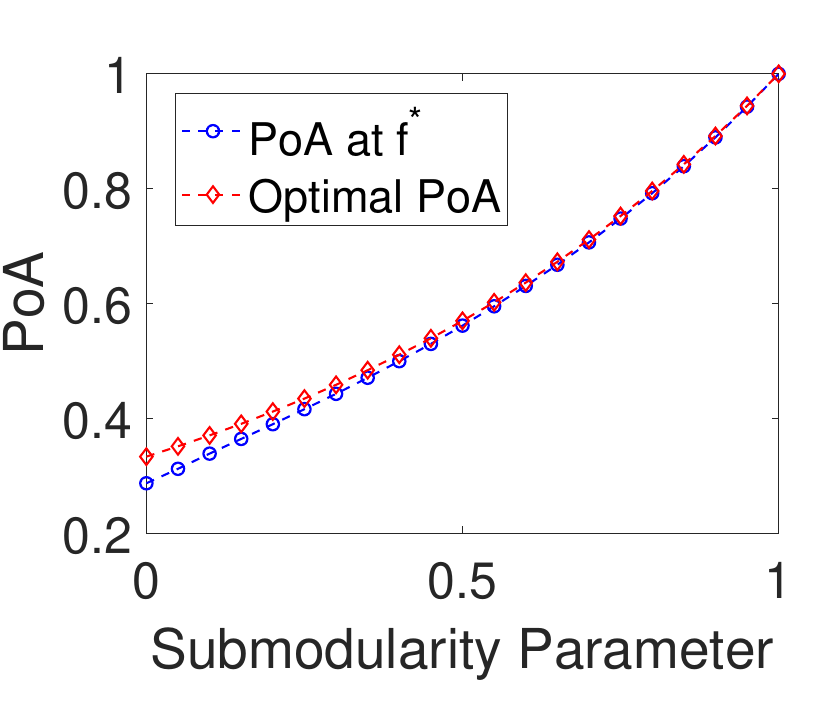}
   \includegraphics[trim ={3.15cm 8cm 4cm 8cm},clip,scale = 0.5]{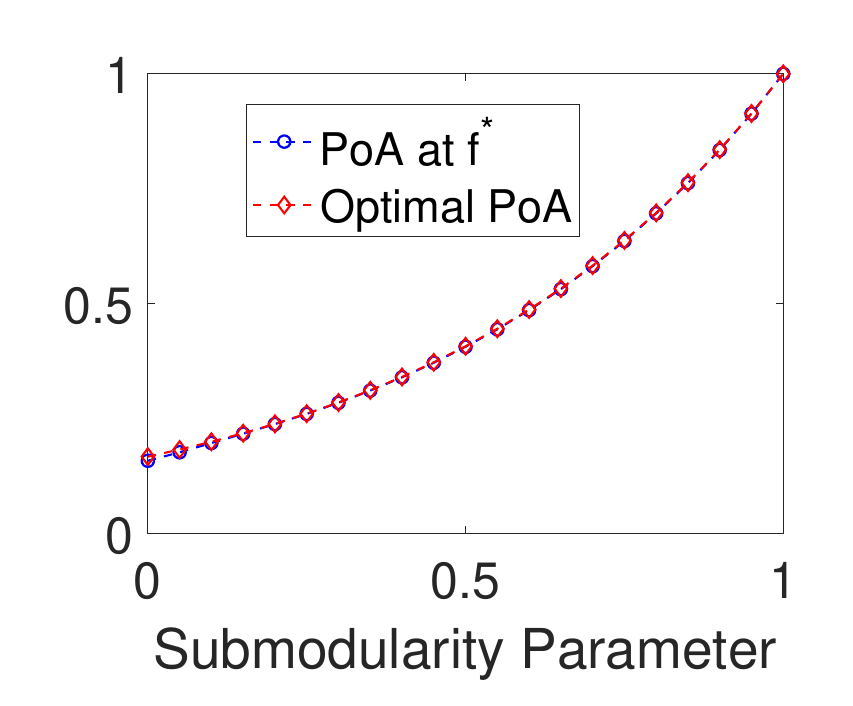}
    \caption{The PoA at optimal utility generating mechanism  $f^*$  corresponding to the full information case where all agents can observe the actions of all other agents converges to the optimal PoA for a network $(\C,\N)$ with $2$ blind and $13$ normal agents as submodularity parameter increases (left). When the number of blind agents increases to $5$, the PoA at $f^*$ and optimal PoA become indistinguishable (right) for any value of submodularity parameter. This depicts the robustness of $f^*$ to communication failure irrespective of the submodular basis function $w$ under consideration.}
    \label{fig_robustness_vs_d}
\end{figure}

Figure~\ref{fig_robustness_vs_bl} also depicts the PoA of the marginal contribution utility; when the submodularity parameter is low, marginal contribution performs very close to the optimal mechanism. 
This appears to be because a low submodularity parameter is ``close'' to a set cover game where marginal contribution has been proved to be optimal. 
When the submodularity parameter is higher, the performance of the marginal contribution utility degrades.
However, with more blind agents, the marginal contribution and $f^*$ converge to the optimal PoA.
Interestingly, the PoA with high submodularity parameter remains more than half even when all the agents become blind as opposed to set cover games, where even one communication failure leads to an optimal PoA of less than half. 

In Figure \ref{fig_robustness_vs_d}, the number of blind agents is fixed to 2 in the left sub-figure and 5 in the right sub-figure, and we plot the PoA at $f^*$ and the optimal PoA for various submodularity parameters $d$. The PoA at $f^*$ and the optimal PoA both improve with the submodularity parameter. 
Further, it is clear that $f^*$ is robust to communication failures irrespective of the submodularity parameter. 
The gap between the optimal PoA and PoA at $f^*$ becomes negligible as more agents in the system become compromised.

Hence, our results indicate that the system designer can assign the $f^*$ utility mechanism to the agents irrespective of the information network, and the PoA remains almost optimal even if there is an unanticipated communication failure.

\begin{figure}
    \includegraphics[trim ={3.6cm 8cm 4cm 8cm},clip,scale = 0.5]{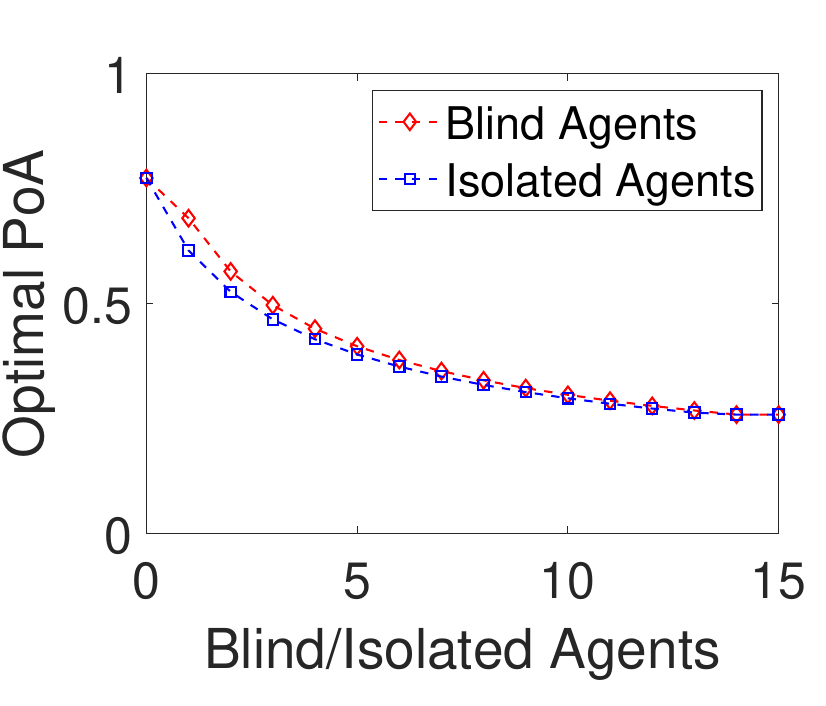}
       \includegraphics[trim ={3.1cm 8cm 4cm 8cm},clip,scale = 0.5]{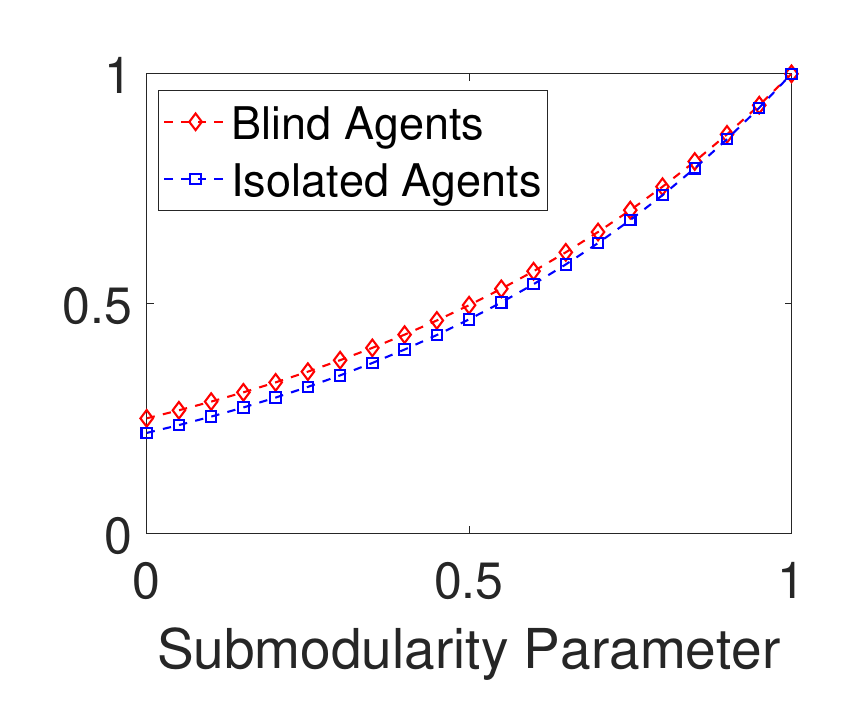}
    \caption{Optimal PoA for a network with $\kappa$ blind agents and $n-\kappa$ normal agents is better than that for a network with $\kappa$ isolated agents and $n-\kappa$ normal agents where $\kappa$ varies from $0$ to $n$  when submodularity parameter is fixed to $d=0.5$ and $n=15$  
    (left). Further, the optimal PoA is monotonically decreasing when more agents become blind or isolated in respective networks. When $\kappa$ is fixed to $3$ and submodularity parameter $d$ is varied in interval $[0,1]$, the network with $\kappa$ blind agents outperforms the network with $\kappa$ isolated agents. This implies that the optimal PoA is monotone with respect to the `connectedness' of the information network among the agents.}
    \label{fig_monotonicity}
\end{figure}

\subsubsection{Monotonicity} Now we study the optimal PoA as the communication level in the network varies; for example, we would say that a system with blind agents is more connected than a system with only isolated agents. 
\revtwo{Figure \ref{fig_monotonicity} gives numerical evidence that at least in the case of blind and isolated agents, when the network is less connected,  the  PoA decreases}. In the left sub-figure, as more agents become compromised and the network becomes less connected, the optimal PoA also degrades. Furthermore, the optimal PoA for a network with blind agents is better than that of a network with isolated agents since the network with blind agents is more connected. In the right sub-figure,  the number of compromised agents is fixed to 3 out of 15 total agents. 
The resulting plot indicates that regardless of the submodularity parameter, the network with blind agents has a better PoA than the network with isolated agents. \revtwo{This suggests that a system with higher connectivity may have better optimal performance than one with lower connectivity; however, this has not been proved and may be an interesting area for future work.}

\section{Arbitrary Games}\label{sec_arbitrary_games}
The previous sections consider the cases of supermodular and submodular games and provide various theorems deriving the \revtwo{pure price of anarchy} for specific information networks. In this section we provide a linear program (LP) in Theorem \ref{thm_dual_lp_gen} that solves \eqref{eqn_poa_gen} for any information network $(\C,\N)$, basis function $w$ and utility generating mechanism chosen by the system designer; this LP is the analytical machinery that underlies all of the technical results from Sections~\ref{sec:supermod} and~\ref{sec:submod}.
Note that the basis function in this section is no longer restricted to submodular or supermodular functions.
\hide{Then we use the LP of Theorem \ref{thm_dual_lp_gen} to develop another LP that derives the exact optimal PoA of \eqref{eqn_poa_opt} and optimal utility generating mechanism for any given network $(\C,\infocN)$ and  basis function $w$ in Theorem~\ref{thm_opt_lp_gen}.}

\subsection{Notations for LP}\label{sec_notations_lp}
We first provide the notation required to present the linear programs. Let $(\C,\N)$ be an arbitrary network with $k$ classes. Recall  $\kappa_j$ represents the number of agents in class $\C_j$, i.e., $\kappa_j = |\C_j|\ge 1$ for $j \in [k]$. Let   $\O_j:=\{l:\C_l\in \N_j\}$ represent the indices of classes whose agents can be observed by agents of $\C_j$; that is,  $\infoc{j}=\{\C_l\}_{l \in \O_j}$ for $j \in [k]$. Let $s_j$ represent the number of agents in $\N_j$, i.e., $s_j = \sum_{l\in\O_j} \kappa_l$.

\revsixteen{The forthcoming linear programs in~\eqref{eqn_poa_primal_lp_gen} (primal) and~\eqref{eqn_poa_dual_lp_gen} (dual) perform a search for pPoA-minimizing game instances parameterized by a given $(f,w,\C,\N)$.
The primal LP's decision variables encode resource values in the underlying game instances; these resources are apportioned to agent action sets in such a way that the intended equilibrium and optimal behavior is realized.
The following notation facilitates these LPs by providing a uniform method of indexing the requisite decision variables and constraints.}

Define a set $\I_\R$ \revsixteen{of \emph{action configurations} as below which will index the constraint set} of the LP deriving the $\poa$ defined in \eqref{eqn_poa_gen}.
%
\revtwo{\begin{eqnarray}\label{eqn_gen_IR}
    \I_\R  \hspace{-2mm}&:= & \hspace{-2mm}\Big\{\big(\gentuple\big): \\
   &&  (a_j,x_j,b_j) \in \mathbb{N}^3_{\ge 0},
     \ 0  \le  a_j+x_j+b_j  \le \kappa_j, \nonumber\\
    &&\hspace{0.5cm} a_jx_jb_j = 0 \mbox{ or } a_j+x_j+b_j = \kappa_j\mbox{ for } j \in [k]\Big\}.\nonumber
\end{eqnarray}}

\revsixteen{Let action configuration $t$ be defined to be $(\ttuple)$ where $t_j=(a_j,x_j,b_j)$; in the formulation of the primal LP, each $t_j$ represents a specific configuration of action choices for agents within class $\C_j$. Then a typical element of $\I_\R$ is $t \in \I_\R$. For any configuration $t \in \I_\R$ define the \emph{equilibrium coverage} $A_t$, the \emph{observable equilibrium coverage} $A_{t,j}$ and the \emph{optimal coverage} $B_t$ as follows:}
\begin{eqnarray}\label{eqn_short}
    A_t &:=& \sum_{j \in [k]} a_j + x_j, \hspace{2mm} B_t\ :=\ \sum_{j \in [k]} b_j + x_j \mbox{ and }\nonumber\\
     A_{t,j} &:=& \sum_{l \in \O_j} a_l + x_l \hspace{2mm} \mbox{ for all } j \in [k].
\end{eqnarray}
 
With all the above notations defined, we are now ready to  present the linear programs.
\begin{table}[h]
\centering
\caption{Summary of Notation}
\revsixteen{\begin{tabular}{ll}
\hline
\textbf{Symbol} & \textbf{Description} \\
\hline
$n$ & Total number of agents \\
$w$ & Basis function\\
$(\C, \mathcal{N})$ & Network with $k$ classes \\
$k$ & Number of classes \\
$\C_j$ & Set of agents in class $j$ \\
$f_j$ & Utility generating mechanism for class $\C_j$\\
$\kappa_j $ & Number of agents in class $j$ \\
$\infoc{j}$ & Set of  classes observable for class $j$ \\
$\mathcal{O}_j$ & Index of classes observable by agents of class $j$ \\
$s_j$ & Total observable agents for class $j$ \\
$\mathcal{I}_R$ & Set of action configurations~\eqref{eqn_gen_IR} \\
$t=(t_1,\dots,t_k)$ & Action configuration; element of $\mathcal{I}_R$
\\
$t_j = (a_j, x_j, b_j)$ & Agent actions in class $j$ at configruation $t$ \\
$A_t$ & Equilibrium coverage at configuration $t$\\
$B_t$ & Optimal coverage at configuration $t$\\
$A_{t,j}$ & Observable equilibrium coverage at configuration $t$\\
\hline
\end{tabular}}
\end{table}
\subsection{LP deriving  PoA for any Utility Generating Mechanism}
The next theorem  provides a linear program that characterizes $\poa$ in \eqref{eqn_poa_gen} for any information network $(\C,\N)$, any basis function $w$ and any  utility generating mechanism $f$ chosen by system designer.
\begin{theorem}[LP to characterize PoA]\label{thm_dual_lp_gen}
For a given network $(\C,\infocN)$, basis function $w$ and utility generating mechanism $f=\{f_j\}_{j\in[k]}$,
\begin{enumerate}[(i)]
 \item If $f_j(1)\le 0$ for any $j \in [k]$, then  $\poa=0$.
\item If $f_j(1)>0$ for all $j \in [k]$, then,
\begin{eqnarray*}
    \poa &=& \frac{1}{\mu^*}
\end{eqnarray*}where $\mu^*$ is the finite solution of the following LP in the unknowns $\mu\in \mathbb{R}$ and $\lambda_j\in \mathbb{R}_{\ge 0}$ for $j\in [k]$:
\begin{eqnarray}\label{eqn_poa_dual_lp_gen}
 \mu^* = && \hspace{-5mm}  \min_{\mu,\lambda_1,\dots,\lambda_k}\ \   \mu \hspace{5mm} \mbox{ subject to }\\
    &&  \hspace{-1.5cm}w(B_t) + \sum_{j \in [k]} \lambda_j  [ a_j f_j(A_{t,j})-  b_j f_j(A_{t,j}+1)]  \le  \mu  w(A_t)\nonumber\\
    && \hspace{4.7cm} \mbox{for all }t  \in \I_\R,\nonumber \\
    &&\hspace{-1.5cm}\mbox{where }f_j(0)= f_j(s_j+1) = w(0)=0 \mbox{ for all }j\in[k].\nonumber 
\end{eqnarray}  
\end{enumerate}
\end{theorem}
\vspace{3mm}

A similar LP is presented in \cite[Theorem 2]{Marden} that derives the PoA for any utility generating mechanism for a network with full information. Theorem \ref{thm_dual_lp_gen} extends that result to cater for any general information network; such an extension is non-trivial and the proof is presented in Appendix \ref{appen_lp}. \revtwo{In Appendix \ref{appen_lp}, an LP is derived that computes the PoA, and the LP in above theorem is the dual LP of the same.}  \revnine{The above LP has has $k+1$ variables and $\prod_{j=1}^k (2\kappa_j^2+2)$ constraints.}

\hide{Theorem \ref{thm_dual_lp_gen} paves the way for optimizing the price of anarchy over all utility generating mechanisms $f$ and solving \eqref{eqn_poa_opt}, which is discussed in the next section.} 

\revtwo{
\begin{remark}[Optimal PoA]\label{remark_opt_Poa_for_bl_is} 
  In an information network with blind, isolated and normal agents, the PNE is guaranteed to exist for all utility generating mechanism. Hence, one can optimize over $f$ to derive the optimal utility generating mechanism and optimal PoA. When $f$ is considered as a decision variable, the terms $\lambda_j  a_j f_j(A_{t,j})$ and $\lambda_j b_j f_j(A_{t,j}+1)$ become non-linear, however,  they can be replaced by another utility generating mechanism $\Tilde{f}_j(\cdot) := \lambda_j f_j(\cdot)$ without changing the PoA by Lemma \ref{lem_scaling_f} making the optimization problem linear again. To do so, note that any feasible $\lambda_j >0$ for all $j$: consider $t\in \I_\R$ with $b_j=1$ and all other components zero, then $\lambda_j \ge f_j(1)$, and by Theorem \ref{thm_dual_lp_gen}  part (i), $f_j(1)>0$ for non-zero PoA (zero PoA cannot be optimal). Then, the resulting LP has $\mu, \Tilde{f}$ as the decision variable. It is easy to extend the arguments in \cite[Lemma 4]{Marden} to show that the resulting LP achieves the optimal solution with bounded components.  
\end{remark}}

\section{Conclusions}

We consider the utility design problem in networked multi-agent coordination using game theory with an aim to derive optimal utility designs and optimal performance guarantees measured by the price of anarchy (PoA) metric. Closed form expressions for optimal utility design leading to optimal PoA for supermodular and submodular system objectives for various networks are derived.
For any general system objective and information network, two linear programs (LP) are provided --- (i) the first LP characterizes the exact PoA for any utility design choice and (ii) the second LP derives the optimal utility design and optimal PoA. 
\revtwo{Our work settles several open questions in utility design for networked multiagent systems (for example, optimal PoA for networked supermodular problems) and highlights several others (for example, extensions such as studying non-submodular problems through the lens of the strong price of anarchy as in~\cite{Ferguson2023}). }

\revtwo{We believe the problem of (non)-existence of pure-strategy Nash equilibrium to be of particular interest.
When networks among agents are nontrivial and PNE may not exist, classical arguments for the relevance of PNE (e.g., potential game-based arguments) lose weight.}
\revnine{In these cases, it may be more relevant to study performance guarantees in games from the standpoint of the limiting behavior of learning rules rather than static equilibrium concepts such as PNE.}
\revsixteen{In such dynamic settings, it will be interesting to see how guarantees compare to the limited static ones which we derive here.}

\revsixteen{Furthermore, future work could investigate whether our LP formulations can be used to prove closed-form bounds for other standard submodular functions, for example  the curvature parameterization used in~\cite{Filmus2014}.}
\revtwo{Finally, it would be interesting to study generally the relationship between fixed utility generating mechanisms such as $f^{\rm mc}$ and optimal mechanisms over all infomation networks.}

\appendices
\section{Proofs of Section \ref{sec_arbitrary_games}}\label{appen_lp}
We first provide another linear program (LP) in Theorem \ref{thm_primal_lp_gen} in the following \revtwo{which} is instrumental in proving Theorem \ref{thm_dual_lp_gen}. 
%
Define the following set which will facilitate the constraints of the LP in Theorem \ref{thm_primal_lp_gen}: 
%
\revtwo{\begin{eqnarray}\label{eqn_gen_I}
    \I &:= &\Big\{(\gentuple) : (a_j,x_j,b_j) \in \mathbb{N}^3_{\ge 0},\nonumber\\
    &&\hspace{1cm} 0 \ \le\  a_j+x_j+b_j \ \le\ \kappa_j \mbox{ for } j\in[k]\Big\}. \hspace{5mm}
\end{eqnarray}}\revtwo{Recall $t$ is defined to be $(\ttuple)$ where $t_j=(a_j,x_j,b_j)$}. Then a typical element of $\I$ is $t \in \I$.  Further, for any $t$ define $A_t$, $A_{t,j}$ and $B_t$ as in \eqref{eqn_short}.

\begin{theorem}[Primal LP to characterize PoA]\label{thm_primal_lp_gen}
For a given network $(\C,\infocN)$, basis function $w$ and utility generating mechanism $f=\{f_j\}_{j\in[k]}$,
\begin{enumerate}[(i)]
\item If $f_j(1)\le 0$ for any $j \in [k]$, then  $\poa=0$.
\item If $f_j(1)>0$ for all $j\in[k]$, then 
\begin{eqnarray}
    \poa &=& \frac{1}{W^*} \label{eq:poa_thm1}
\end{eqnarray}where $W^*$ is the finite \revtwo{value} of the following LP in the unknowns $\{\theta(t)\}$:
\end{enumerate}
\begin{eqnarray}\label{eqn_poa_primal_lp_gen}
   W^* &=& \max_{\theta(t)}\sum_{t\in \I} w(B_t) \theta(t) \hspace{5mm} \mbox{subject to} \nonumber\\
&& \hspace{-10mm} \sum_{t\in \I} [ a_j f_j(A_{t,j})-  b_j f_j(A_{t,j}+1)] \theta(t) \ge 0,\hspace{2mm} j \in [k]\nonumber\\
    && \hspace{-10mm}\sum_{t\in\I} w(A_t)\theta(t)=1,\nonumber\\
    && \hspace{-10mm} \theta(t) \ge 0 \ \hspace{2mm} \mbox{for all}\ \ t\in \I, \\
  \mbox{and} && \hspace{-10mm}\ f_j(0)= f_j(s_j+1) = w(0) =0, \ \forall \ j\in[k].\nonumber
\end{eqnarray}

\end{theorem}
\vspace{2mm}

\textbf{Proof:} \revtwo{Part (i) follows by \cite[Lemma 1]{Marden}.} Towards part (ii), we first come up with a reduced class of games \revtwo{$\hGG\subseteq \GG$} such that PoA over class of games $\hGG$ is same as that over the class of games $\GG$. Then we re-write \eqref{eqn_poa_gen} as a constrained optimization problem  in \eqref{eqn_gen_org_ot}. Next, we consider a relaxation of \eqref{eqn_gen_org_ot} in \eqref{eqn_gen_relax_ot} and show the equivalence of \eqref{eqn_gen_relax_ot} and \eqref{eqn_poa_primal_lp_gen}. Lastly, using Lemma \ref{lem_gen_relaxation_equiv} we show that the relaxation is indeed tight, and  optimization problems \eqref{eqn_gen_org_ot} and \eqref{eqn_gen_relax_ot} are equivalent which completes the proof.

Let $\hatGG$ be the set of games that contains one  $\hat{G}$ for each $G \in \GG$ such that $\hat{G}$ matches $G$ in everything except the action sets of the agents.  The action set of any agent $i$ in $\hat{G}$ only contains the action at the worst NE of $G$ and that at an optimal action profile, that is, $\A_i = \{\ane_i,\aopt_i\}$. \revtwo{By Lemma \ref{lemStep2OfMarden},} the price of anarchy  $\poa= \frac{1}{W^*}$, where $W^*$ is the solution of following optimization problem:
\begin{eqnarray}\label{eqn_gen_org_ot}
    W^* &=& \sup_{\hat{G}\in \hGG} W(\aopt)\nonumber\\
 \mbox{s.t.} && U_i(\ane) \ge U_i(\aopt_{i},\ane_{-i})\hspace{5mm} \forall \ \ i \in N,\nonumber\\
&& W(\ane)=1. 
\end{eqnarray}To construct a linear program solving optimization problem in \eqref{eqn_gen_org_ot} consider the following relaxation:
\begin{eqnarray}\label{eqn_gen_relax_ot}
    V^* &=& \sup_{\hat{G}\in \GG} W(\aopt)\nonumber\\
 \mbox{s.t.}&& \sum_{i \in \C_j}\left[ U_i(\ane) - U_i(\aopt_{i},\ane_{-i})\right] \ge 0,\hspace{5mm} \mbox{for } j \in [k],\nonumber\\
&& W(\ane)=1.
\end{eqnarray}Further, let
\begin{enumerate}[(a)]
    \item $x_{j,r} \in \{0,[\kappa_j]\}$ be the number of agents in $\C_j$ selecting resource $r$ in both $\ane$ and $\aopt$. 
    \item $a_{j,r} + x_{j,r} \in \{0,[\kappa_j]\}$ be the number of agents in $\C_j$ selecting resource $r$ in $\ane$.
    \item $b_{j,r}  + x_{j,r}  \in \{0,[\kappa_j]\}$ be the number of agents in $\C_j$ selecting resource $r$ in $\aopt$.
    \item  $A_{r}$ be total number of agents who select resource $r$ at $\ane$, thus $A_{r}:= \sum_{j \in [k]} a_{j,r} + x_{j,r}$.
    \item  $A_{r,j}$ be the number of agents from classes $\O_j$ who select resource $r$ at $\ane$, thus $A_{r,j} := \sum_{l \in \O_j} a_{l,r} + x_{l,r}$.
    \item  $B_{r}$ be total number of agents who select resource $r$ at $\aopt$ thus  $B_{r} := \sum_{j \in [k]} b_{j,r} + x_{j,r}$.
\end{enumerate}For each $t \in \I$ of \eqref{eqn_gen_I}, define $\R(t)$ as follows:
\begin{eqnarray}
    \R(t)&=& \{r \in \R: a_{j,r}+x_{j,r} = a_j+x_j,\ x_{j,r} = x_j\nonumber \\
    &&\hspace{10mm}b_{j,r}+x_{j,r} = b_j+x_j \mbox{ for } j\in[k]\}\label{eqn_Rt}
\end{eqnarray}to be the set of all the resources that are selected by exactly $a_j+x_j$ agents from $\C_j$ at $\ane$, $b_j+x_j$  at $\aopt$ and $x_j$ agents in both $\ane$ and $\aopt$. Further define $\theta(t)$ to be the sum of the values of the resources in $\R(t)$, i.e., $\theta(t) = \sum_{r \in \R(t)} v_r$. We can use $\{\theta\}$ to define the quantities in  \eqref{eqn_gen_relax_ot}. Let $A_t,B_t,A_{t,j}$ be as in \eqref{eqn_short}, then
\begin{eqnarray}\label{eqn_w_ane}
    W(\ane) &=&  \sum_{r\in \R} v_r \ w(A_r)\\
            &=& \sum_{t\in\I} w(A_t) \sum_{r \in \R(t)} v_r = \sum_{t\in\I} w(A_t) \theta(t).\nonumber
\end{eqnarray}Using similar logic, we get
\begin{eqnarray}\label{eqn_w_aopt}
    W(\aopt) &=& \sum_{t\in\I} w(B_t) \theta(t) \hspace{5mm} \mbox{ and } \\
    \sum_{i \in \C_j} U_i(\ane) 
                                &=& \sum_{t\in\I} (a_j+x_j)f_j(A_{t,j})\theta(t).\label{eqn_U_ane}
\end{eqnarray}Lastly, $\sum_{i \in \C_j} U_i(\aopt_{i},\ane_{-i})$ equals
\begin{eqnarray}\label{eqn_U_ai_ane}
&&\hspace{-7mm}= \sum_{r \in \R} v_r[ x_{j,r}  f_j(A_{r,j})+ b_{j,r}  f_j(A_{r,j}+1)]\nonumber\\
&&\hspace{-7mm}= \sum_{t \in \I} [ x_j f_j(A_{t,j})+  b_j  f_j(A_{t,j}+1)] \theta(t).
\end{eqnarray}Thus, the optimization problem \eqref{eqn_gen_relax_ot} is equivalent to
\begin{eqnarray}\label{eqn_LP_with_sup}
   W^* &=& \sup_{\theta(t)}\sum_{t\in \I} w(B_t) \theta(t) \nonumber\\
\mbox{s.t.}\hspace{5mm}&& \hspace{-10mm} \sum_{t\in \I} [ a_j f_j(A_{t,j})-  b_j f_j(A_{t,j}+1)] \theta(t) \ge 0,\hspace{1mm} j \in [k],\nonumber\\
    && \hspace{-10mm}\sum_{t\in\I} w(A_t)\theta(t)=1,\nonumber\\
    && \hspace{-10mm} \theta(t) \ge 0 \ \hspace{5mm} \mbox{ for all }\ \ t\in \I, \\
  \mbox{and} && \hspace{-10mm}\ f_j(0)= f_j(s_j+1) = w(0) =0, \ \forall \ j\in[k].\nonumber
\end{eqnarray}To see that the maximum in above is achieved, note that the objective function is continuous in $\{\theta\}$. All $\theta(t)\ge 0$ are bounded below. The constraint on $W(\ane)=1$ implies that $\theta(t)$ is bounded for all $t \in \I$ with $A_{t} \ge 1$. For $t \in \I$ such that $A_t= 0$, the equilibrium constraint for class $\C_j$ can be re-written as,
\begin{eqnarray*}
    \sum_{\substack{t\in\I\\ A_t=0}} b_j f_j(1) \theta(t) &\le& \sum_{\substack{t\in\I\\ A_t\ge 1}} [ a_j f_j(A_{t,j}-  b_j f_j(A_{t,j}+1)] \theta(t)
\end{eqnarray*}which provides the required boundedness as $f_j(1) >0$. Thus  $\{\theta(t)\}$ belong to a compact domain and optimization problem in \eqref{eqn_gen_relax_ot} is equivalent to the linear program in \eqref{eqn_poa_primal_lp_gen}.

Clearly, $V^* \ge W^*$ since any feasible point for optimization problem \eqref{eqn_gen_org_ot} is also a feasible point for optimization problem \eqref{eqn_gen_relax_ot} and hence for \eqref{eqn_poa_primal_lp_gen}. Lemma \ref{lem_gen_relaxation_equiv} (presented next) shows that $V^* \le W^*$, obtaining the proof of Theorem~\ref{thm_primal_lp_gen}. \eop
\revtwo{
\begin{lemma}\label{lemStep2OfMarden}
     Let $\hatGG$ be the set of games that contains one  $\hat{G}$ for each $G \in \GG$ such that $\hat{G}$ matches $G$ in everything except the action sets of the agents.  The action set of any agent $i$ in $\hat{G}$ only contains the action at the worst NE of $G$ and that at an optimal action profile, that is, $\A_i = \{\ane_i,\aopt_i\}$. Then price of anarchy  $\poa$ of \eqref{eqn_poa_gen} equals $\frac{1}{W^*}$, where $W^*$ is the solution of following optimization problem:
\begin{eqnarray}\label{lem_eqn_gen_org_ot_inv}
    W^* &=& \inf_{\hat{G}\in \hGG}\frac{1}{W(\aopt)}\nonumber\\
 \mbox{s.t.} && U_i(\ane) \ge U_i(\aopt_{i},\ane_{-i})\hspace{5mm} \forall \ \ i \in N,\nonumber\\
&& W(\ane)=1. 
\end{eqnarray}
\end{lemma}
\noindent\textbf{Proof:} The proof follows in the exact same manner as in step 2 of the proof of \cite[Theorem 2]{Marden}. We reproduce the proof here for the sake of completion.}

\revtwo{Using Lemma \ref{lemStep1OfMarden}, it is sufficient to consider the games in $\hGG$ to derive the $\poa$. By \cite[Lemma 2]{Marden}, for every game $G\in \GG$ and hence for every game $\hat{G}\in \hGG$,  the system objective $W(\ane)>0$. Now for a fixed game $\hat{G}$, create another game $\Tilde{G}$ that matches $\hat{G}$ in every aspect except the value of resources. Each resource in game $\hat{G}$ is scaled by $W(\ane)$ of $\hat{G}$ in the game $\Tilde{G}$. Basically, a resource $r$ with value $v_r$ in game $\hat{G}$ now has value of $\frac{v_r}{W(\ane)}$ in game $\Tilde{G}$. With this scaling, the worst performing equilibrium and optimal action profile remain unchanged. However, the system objective at worst equilibrium of $\Tilde{G}$, $W(\ane)$, equals 1 using \eqref{eqn_syst_obj} while the PoA of $\hat{G}$ and $\Tilde{G}$ remain equal. Further, $\{\hat{G}\in \hGG: W(\ane)=1\}\subseteq \hGG$. Hence, following similar arguments as  in proof of Lemma \ref{lemStep1OfMarden}, it is clear that the PoA over class $\hGG$ is the same as the subclass that is constrained to $W(\ane)=1$. Thus the PoA in \eqref{eqn_poa_gen} can be computed using \eqref{lem_eqn_gen_org_ot_inv}. This completes the proof. \eop }
\revtwo{
\begin{lemma}\label{lemStep1OfMarden}
 The price of anarchy over $\hatGG$ of Lemma \ref{lemStep2OfMarden} is same as that over $\GG$, that is,
 \begin{eqnarray*}
    \poa= \inf_{G \in \hGG}\left(\frac{\min_{\aa \in {\rm PNE}(G)} W(\aa)}{\max_{\aa \in \A}W(\aa)}\right).
 \end{eqnarray*}
\end{lemma}
\noindent\textbf{Proof:} The proof follows in the exact same manner as in step 1 of the proof of \cite[Theorem 2]{Marden}. We reproduce the proof here for the sake of completion.}

\revtwo{By definition $\hGG \subseteq \GG$, and for any game $G \in \GG$,  there is a game $\hat{G} \in \hGG$ such that $G$ and $\hat{G}$ have the same pure price of anarchy,
\begin{eqnarray*}
    \frac{\min_{\aa \in {\rm PNE}(\hat{G})} W(\aa)}{\max_{\aa \in \A}W(\aa)} = \frac{\min_{\aa \in {\rm PNE}(G)} W(\aa)}{\max_{\aa \in \A}W(\aa)}.
\end{eqnarray*}Hence the proof. \eop
}
\begin{lemma}\label{lem_gen_relaxation_equiv}
    Consider $W^*$ as in \eqref{eqn_gen_org_ot} and $V^*$ as in \revtwo{\eqref{eqn_gen_relax_ot}}. It holds that  $V^* \le W^*$. 
\end{lemma}
\noindent\textbf{Proof:} Let  $\{\theta(t)\}$ be any feasible solution of LP \eqref{eqn_poa_primal_lp_gen} with objective value $V$. We construct a game instance $G$ that satisfies the constraints of \eqref{eqn_gen_org_ot} with the same objective value. This will conclude $V^* \le W^*$. 

Let $\n := \prod_{j=1}^k \kappa_j$. \revtwo{Let ${\cal R}=\I\times [\n]$.  For each $(t,q)\in {\cal R}$, construct  resources $r(t,q)$ with value $v_{r(t,q)}=\frac{\theta(t)}{\n}$}. For each $j\in [k]$, $t \in \I$ \revtwo{ and $l \in [\kappa_j]$} define the following sets:
\revtwo{\begin{eqnarray*}
 K_j(t,l) &=&  \left\{r_{(t,q)}: q=(l-1) \frac{\n}{\kappa_j}+1,\dots,l \frac{\n}{\kappa_j} \right\} \\
\end{eqnarray*}}%
\revtwo{Add sets $ K_j(t,l)$ to the action sets of agents in class $\C_j$ in the following manner. For fixed $t$, arrange $\kappa_j$ sets $K_j(t,l)$ in a circle indexed by $l \in[\kappa_j]$  in clockwise direction. Then, order the agents in class $\C_j$ from $1$ to $\kappa_j$. For the agent  in $i$-th position, add $a_j+x_j$ sets to the equilibrium action $\ane_{ij}$  starting from set  $K_j(t,i)$ moving clockwise. For optimal action $\aopt_{ij}$, add $b_j+x_j$ sets starting from $K_j(t,1+(i-1-b_j)\mod \kappa_j)$ moving clockwise\footnote{With slight abuse of notation, we use $\ane_{ij}$ and $\aopt_{ij}$ to represent the action set of player at the $i$-th position in $\C_j$, whereas this player may have been assigned a different label from $[n]$ in the original problem.}. Formally,}
\revtwo{\begin{eqnarray*}
 &&\hspace{-7mm} \ane_{ij}=
      \{ K_j(t,l) : a_j+x_j \ge 1+ ((l-i)\hspace{-3mm}\mod \kappa_j), l \in [\kappa_j]\},  \\ 
   &&\hspace{-7mm}   \aopt_{ij} \hspace{-2mm}=
     \{ K_j(t,l) : b_j+x_j \ge 1+ (l-i+b_j)\hspace{-3mm}\mod \kappa_j,  l \in [\kappa_j]\} .  
\end{eqnarray*}}\revtwo{ This process is repeated for all $t$, all $i\in \C_j$ and all $\C_j$. For example, consider an instance with $a_j=2$, $x_j=2$, $b_j=1$, $\kappa_j=6$ and $i=2$ for some $\C_j$. Then $ a_j+x_j \ge 1+ ((l-i)\mod \kappa_j)$ is true only for $l=2,3,4$ and $5$, and $b_j+x_j \ge 1+ (l-i+b_j)\mod \kappa_j$ is true only for $l=1,2,3$. Hence, agent $i$ selects $a_j+x_j=4$ sets of resources at equilibrium, $b_j+x_j=3$ sets of resources at optimal action; and resources in $x_j=2$ sets $K_j(t,2)$ and $K_j(t,3)$ are common in equilibrium and optimal actions.} \revtwo{It is not difficult to show that this construction} results in a game where for all $t \in \I$,
\begin{enumerate}[\textbf{B}.1]
\item  Any resource $r_{(t,q)}$ for $q \in[\n]$ has a value of $\frac{\theta(t)}{\n}$.
    \item  Any  $r_{(t,q)}$ for $q \in[\n]$ is selected by exactly $a_j+x_j$  agents from $\C_j$ at $\ane$, thus $A_t$ agents (see \eqref{eqn_short}). Further, $r_{(t,q)}$ is selected by $b_j+x_j$ agents from class $\C_j$ at $\aopt$ (thus $B_t$ agents).
    \item Each agent in $\C_j$ selects $(a_j+x_j)$ sets out of $\{K_j(t,l)\}_{l \in [\kappa_j]}$ at $\ane$, each of which contains $\frac{\n}{\kappa_j}$ resources. Thus every agent selects $\frac{\n}{\kappa_j}(a_j+x_j)$ resources at $\ane$. Every agent selects $b_j+x_j$ such sets  at $\aopt$, thus selects $\frac{\n}{\kappa_j}(b_j+x_j)$ resources. Further, $x_j$ sets, and hence $\frac{\n}{\kappa_j} x_j$ resources are common in $\ane$ and $\aopt$.

\end{enumerate}%
We now prove that $W(\ane)=1$ and $W(\aopt)=V$. Observe that 
\revtwo{\begin{eqnarray}\label{eqn_wane_1}
    W(\ane) = \hspace{-1mm}\sum_{t \in\I}  \sum_{q\in [\n]}\hspace{-1mm} v_{r(t,q)} =\sum_{t\in\I} \theta(t) w(A_t)=1\hspace{2mm}
\end{eqnarray}}by the constraint in the LP in \eqref{eqn_poa_primal_lp_gen}. Similar arguments prove that $W(\aopt)=V$. 


For any $t \in \I$, the resource $r_{(t,q)}$ is selected by $A_{t,j}= \sum_{p \in \O_j} a_{p}+x_{p}$ agents from $\N_j$ for $j \in [k]$. The function \eqref{eqn_G_gen} of Lemma \ref{lem_G_func} at $\ane$ equals,
\begin{equation}\label{eqn_G_eq}
      G_j(\ane)=  \sum_{q \in [\n]} \sum_{t\in\I}  \frac{\theta(t)}{{\n}}\sum_{l=1}^{A_{t,j}}f_j(l) =\frac{1}{{\n}}\sum_{t\in\I} {\n}  \theta(t)\sum_{l=1}^{A_{t,j}}f_j(l).
\end{equation}When an agent $i \in \C_j$ unilaterally deviates from $\ane$ to $(\aopt_i,\ane_{-i})$, then each $K_j(t,l) \in \ane_i$ which is not in $\aopt_i$ would be selected one less time, and each $K_j(t,l) \in \aopt_i$ which is not in $\ane_i$ would be selected one more time.  Thus,  $a_j$ sets of the type $K_j(t,l)$ get selected by one less agent and $b_j$ sets of the type $K_j(t,l)$ get selected by one extra agent. Thus for each $t \in\I$, $\frac{\n}{\kappa_j}a_j$ resources get selected by $A_{t,j} -1$ agents from $\N_j$ and $\frac{\n}{\kappa_j}b_j$ resources get selected by $A_{t,j}+1$ agents from $\N_j$, rest are selected by $A_{t,j}$ agents from $\N_j$. Then $G_j(\aopt_i,\ane_{-i})$ equals,
\begin{eqnarray}\label{eqn_G_devi}
    &&\hspace{-0.75cm}=\ \frac{1}{{\n}}\sum_{t\in\I}  \theta(t)\Big[\frac{\n b_j}{\kappa_j}\sum_{l=1}^{A_{t,j}+1}f_j(l)+\frac{\n a_j}{\kappa_j} \sum_{l=1}^{A_{t,j}-1}f_j(l)\nonumber\\
    &&\hspace{2.5cm}+\Big({\n}-\frac{\n a_j}{\kappa_j} -\frac{\n b_j}{\kappa_j}\Big)\sum_{l=1}^{A_{t,j}}f_j(l)\Big]\\
   &&\hspace{-0.75cm}=\ \frac{1}{{\n}}\sum_{t\in\I}  \theta(t)\Big[\frac{\n b_jf_j(A_{t,j}+1)}{\kappa_j} - \frac{\n a_j f_j(A_{t,j})}{\kappa_j}+{\n}\sum_{l=1}^{A_{t,j}}f_j(l)\Big].\nonumber
\end{eqnarray}From \eqref{eqn_G_eq} and \eqref{eqn_G_devi}, $ G_j(\ane)-G_j(\aopt_i,\ane_{-i}) $ equals
\begin{eqnarray}\label{eqn_proving_eq}
    &&\hspace{-0.75cm} =\frac{1}{{\kappa_j}}\sum_{t\in\I}  \theta(t) [a_j f_j(A_{t,j}) -b_jf_j(A_{t,j}+1)] \ \ge\ 0
\end{eqnarray}by the constraint in the LP \eqref{eqn_poa_primal_lp_gen}.  Similar arguments follow for all $j\in[k]$. By Lemma \ref{lem_G_func} (presented next), $\ane$ is a Nash equilibrium, hence satisfying the constraints in \eqref{eqn_gen_org_ot}.  \eop

\begin{lemma}\label{lem_G_func}
For any $j \in [k]$, let $G_j(\cdot)$ be defined as
\begin{eqnarray}\label{eqn_G_gen}
     G_j(\aa) &=& \sum_{r \in \cup_i a_i} v_r  \sum_{l=1}^{|\aa|_r^{\N_j}} f_j(l).
\end{eqnarray}Then, $  G_j(\aa) -G_j(b_i,a_{-i}) = U_i(\aa)-U_i(b_i,a_{-i})$  for all  $ b_i \in \A_i$, $i \in \C_j $ and all $j\in[k]$.
\end{lemma}
\noindent\textbf{Proof:} When an agent $i\in \C_j$ deviates to $b_i = \{a_i \backslash R_1\}\cup R_2$ where $R_1,R_2 \subset \R$, then 
\begin{eqnarray*}
   U_i(b_i,a_{-i}) = \sum_{r \in  a_i\backslash R_1} v_r f_j(|\aa|^{\N_j}_r) + \sum_{r \in  R_2} v_r f_j(|\aa|^{\N_j}_r+1)\\
                   = U_i(a_i,a_{-i}) -\sum_{r \in R_1} v_r f_j(|\aa|^{\N_j}_r) + \sum_{r \in  R_2} v_r f_j(|\aa|^{\N_j}_r+1).
\end{eqnarray*}Further, $G_j(b_i,a_{-i})$ equals,
\begin{eqnarray*}
    &=& G_j (\aa)-\sum_{r \in R_1} v_r f_j(|\aa|^{\N_j}_r) + \sum_{r \in  R_2} v_r f_j(|\aa|^{\N_j}_r+1).
\end{eqnarray*}Hence the proof.  \eop

\subsubsection*{Proof of Theorem \ref{thm_dual_lp_gen}} The dual of the LP in \eqref{eqn_poa_primal_lp_gen} equals (see  e.g.,\cite{Bertsimas})
\revtwo{\begin{eqnarray}\label{eqn_poa_dual_lp_gen_temp}
\mu^* & =&  \min_{\lambda_1\ge0,\dots,\lambda_k\ge 0,\mu\in \mathbb{R}} \ \   \mu \hspace{5mm} \mbox{ subject to }\\
    && \hspace{-15mm} w(B_t) + \sum_{j \in [k]} \lambda_j  [ a_j f_j(A_{t,j})-  b_j f_j(A_{t,j}+1)]  \le  \mu  w(A_t)\nonumber \\
    && \hspace{5cm}\ t  \in \I_\R,\nonumber \\
    && f_j(0)= f_j(s_j+1) = w(0) =0,  \forall \ j\in[k]. \nonumber
\end{eqnarray}} By strong duality  (\cite{Bertsimas}), $\mu^*$ of above equals $W^*$ of \eqref{eqn_poa_primal_lp_gen}. It remains to show that considering all $t \in \I$ of \eqref{eqn_gen_I} is equivalent to considering all $t \in \I_\R$ of \eqref{eqn_gen_IR}. 

Observe that $t \in \I$ such that $A_t=0$ or $B_t=0$ are all contained in $\I_\R$. Next, we show that the constraints corresponding to $t\in \I$  such that $A_t \ne 0$ and $B_t \ne 0$ are equivalent to the constraints corresponding to  $t\in \I_\R$. 

\revtwo{For any $t\in \I$ such that $A_t \ne 0$, $B_t \ne 0$, consider $t'$ such that $a_j'+x_j' = a_j+ x_j$ and $b_j'+x_j' = b_j+x_j$ for all $j$.}
Consider a change of coordinates from $t$ to $\Tilde{t}=(\tltuple)$ where $\Tilde{t}_j$ equals $(l_j,x_j,m_j)$ with $l_j := a_j+x_j$ and $m_j := b_j+x_j$. Then the constraint in \eqref{eqn_poa_dual_lp_gen_temp} can be re-written as
\begin{eqnarray}\label{eqn_gen_cons}
   \mu  w(A_t)
                             &\ge&w(B_t)  + \sum_{j\in[k]} \lambda_j  [l_j f_j(A_{t,j})- m_j f_j(A_{t,j}+1)\nonumber\\
                             &&\hspace{1.5cm}+ x_j[f_j(A_{t,j}+1) -f_j(A_{t,j})]] \hspace{2mm}
\end{eqnarray}for all $\Tilde{t} \in \hat{\I}$, where $  \hat{\I}$ equals
\begin{eqnarray*}
   &&\hspace{-0.75cm}\Big\{(\tltuple) : (l_j,x_j,m_j) \in \mathbb{N}^3_{\ge 0},\ 0\le l_j-x_j+m_j\le\kappa_j, \\
   &&l_j \ge x_j,\  m_j \ge x_j, \  \sum_{j \in[k]}l_j \ne 0,\  \sum_{j \in[k]}m_j \ne 0\nonumber \\
    &&\hspace{2mm}\mbox{ and }\ 1 \ \le\ \sum_{j \in[k]} (l_j-x_j+m_j)  \le \ n\Big\}.
\end{eqnarray*}In the remainder of this proof, fix $l_j$ for all $j$, and let $x_j$ and $m_j$ move freely in $\hat{\I}$. When $l_j = \kappa_j$ for all $j\in[k]$, then $m_j = x_j$ since $m_j-x_j \ge 0$ and $-x_j+m_j \le 0$.  This implies $b_j = 0$ for all $j$, and note that this point is also contained in $\I_\R$. Now consider the $\Tilde{t} \in\hat{\I}$ such that $l_j \ne \kappa_j$ for at least one $j\in[k]$, and define two sets $\R_1$ and $\R_2$ as,
\begin{eqnarray*}
    \R_1 &=& \{j \in [k]:\ f_j(A_{t,j}+1) -f_j(A_{t,j}) \le 0\},\\
    \R_2 &=& \{j \in [k]:\ f_j(A_{t,j}+1) -f_j(A_{t,j}) > 0\}.
\end{eqnarray*}For any fixed $l_j,m_j$ for $j\in [k]$, the most binding constraint in \eqref{eqn_gen_cons} arises when $x_j$ is picked as small as possible for $j \in \R_1$ (since $x_j \ge 0$ and is multiplied with a negative coefficient), and $x_j$ is picked as large as possible for $j \in \R_2$ (since $x_j \ge 0$ and is multiplied with a positive coefficient). 

For $j \in \R_1$, since $x_j \ge l_j + m_j - \kappa_j$, for fixed $m_j$ and $l_j$, the smallest value $x_j$ can take is $x_j =\max\{0,l_j+m_j-\kappa_j\}$. There are two possibilities,
\begin{enumerate}[\textbf{B}.1]
    \item when $l_j+m_j \le \kappa_j$, that is when $a_j+b_j+2x_j \le n$, then $x_j =0$ implying $a_jx_jb_j =0$,
    \item when $l_j+m_j > \kappa_j$, that is when $a_j+b_j+2x_j >\kappa_j$ then $x_j = l_j+m_j - \kappa_j$ implying $a_j+x_j+b_j = \kappa_j$.
\end{enumerate}Similarly, for $j \in \R_2$, since $x_j \le l_j $ and $x_j \le m_j$, for fixed $m_j$ and $l_j$ largest value of $x_j$ is $x =\min\{l_j,m_j\}$. Again, there are two possibilities,
\begin{enumerate}[\textbf{B}.1]
\setcounter{enumi}{2}
    \item when $l_j \le m_j$, that is when $a_j\le b_j$ then $x_j =l_j$ and thus $a_j =0$ implying $a_jx_jb_j =0$,
    \item when $l_j > m_j$, that is when $a_j > b_j$ then $x_j =m_j$ and thus $b_j =0$ implying $a_jx_jb_j =0$.
\end{enumerate}In all, for any $j \in [k]$, exactly one of the \textbf{B}.1-\textbf{B}.4 holds at the most binding constraint. Since all such quantities are included in $\I_\R$, it is sufficient to consider $\I_\R$ and LP in \eqref{eqn_poa_dual_lp_gen_temp} is equivalent to the LP \eqref{eqn_poa_dual_lp_gen}. \eop 


\begin{lemma}\label{lem_scaling_f}
    Consider any $\alpha_j>0$ for $j\in [k]$ and any utility generating mechanism $f = \{f_j\}_{j\in[k]}$. Define another utility generating mechanism $f^\alpha = \{f^\alpha_j\}_{j\in [k]}$ as $ f^\alpha_j(x) = \alpha.f_j(x)$ for all $j\in [k]$. Then for any $w$ and network $(\C,\N)$,
    \begin{equation}
      {\rm PoA}(f^\alpha,w,\C,\N) = \poa.
    \end{equation}
\end{lemma}
\noindent\textbf{Proof:} The proof immediately follows from the constraint in (primal) LP \eqref{eqn_poa_primal_lp_gen} in Theorem \ref{thm_primal_lp_gen}. \eop


\section{} \label{appen_remaining_proofs}
To follow the proofs in this section, one needs familiarize themselves with \revother{the result of Theorem \ref{thm_dual_lp_gen}, although the proof of Theorem \ref{thm_dual_lp_gen} is not required in this section.} The quantities defined in Section \ref{sec_notations_lp} are also used in most of the proofs presented in the following.

\subsubsection*{Proof of Theorem \ref{thm_sup_mod}} By Lemma \ref{lem_scaling_f} in Appendix \ref{appen_lp}, we can restrict our attention to utility generating mechanism $f$ such that $f_j(1)=1$ for all $j \in [k]$.


\revtwo{Consider a $t\in \I_\R$ defined in \eqref{eqn_gen_IR} with $A_t = 0$  and $B_t \ne 0$. Such a $t$ implies that $\sum_{j \in [k]} b_j \ge 1$ by \eqref{eqn_short}. The constraint in LP \eqref{eqn_poa_dual_lp_gen} deriving $\poa$ at this $t$  equals
\begin{eqnarray}\label{eqn_lmd_w_rel_sup_mod}
    \sum_{j \in [k]} b_j \lambda_j  & \ge & w \Big( \sum_{j \in [k]}b_j\Big) \mbox{ for all } b_j \in \{0,[\kappa_j]\}.
\end{eqnarray} If $t$ is such that $b_j =\kappa_j$, the above constraint becomes
\begin{eqnarray}\label{eqn_cons_lmd_sup_mod}
    \sum_{j \in [k]} \kappa_j \lambda_j  & \ge & w (n).
\end{eqnarray}Observe that in the above, equality is achieved when $\lambda_j =w(n)/n$ for all $j$. When $k=1$, $\lambda_1\ge w(n)/n$. When $k\ge 2$, if $\lambda_l< w(n)/n$, then there must be some $\lambda_p>w(n)/n$ for the constraint to be true for $l,p \in [k]$.
Further, when $t \in \I_\R$ is such that $A_t \ne 0$ and $B_t =0$, define $\Tilde{A}_{t,j} = A_{t,j}|_{x_j =0}$, the constraint in \eqref{eqn_poa_dual_lp_gen} at this $t$ equals
\begin{equation}\label{eqn_lb_mu_sup_mod1}
    \mu\ge  \frac{1}{w\big(\sum_{j \in[k]} a_j\big)} \sum_{j \in [k]} \lambda_j a_j f(\Tilde{A}_{t,j})\mbox{ for all } a_j \in \{0,[\kappa_j]\}. \hspace{2mm}
\end{equation}
The above constraint for a $t$ with $a_j =1$ and $a_l =0$ for $l \in [k]\backslash j$ implies $\mu \ge \lambda_j$ for all $j$. Hence, 
\begin{eqnarray}\label{eqn_lb_mu_sup_mod}
    \mu & \ge &\frac{w(n)}{n} \hspace{5mm}\mbox{ thus}\\
     \poa &\le&  \frac{1}{\mu}\ \le\ \frac{n}{w(n)} \mbox{ for all } f.\label{eq:upbd_sup_mod_poa}
\end{eqnarray}}
\revtwo{Now consider the communication denied utility generating mechanism of \eqref{eqn_def_CD}; we prove that the upper bound on $\poa$ in \eqref{eq:upbd_sup_mod_poa} is achieved.}

\revtwo{Consider a $t \in \I_\R$ such that $A_t \ne 0$. 
We first prove that the constraints with $B_t \ne 0$ are implied by those with $B_t =0$. The LP constraint \eqref{eqn_poa_dual_lp_gen} when $A_t \ne 0$ and $B_t \ne 0$ equals\begin{eqnarray}\label{eqn_cons_At_bt_sup_mod}
    \mu &\ge& \frac{1}{w(A_t)} \Big(w(B_t) - \sum_{j \in [k]} \lambda_j b_j \nonumber \\
    && \hspace{3cm}+ \sum_{j \in [k]} \lambda_j a_j \indc{A_{t,j}\ne 0}\Big).
\end{eqnarray}Observe that the RHS of the above constraint is less than or equal to following: 
\begin{eqnarray}
    &&\hspace{-7mm}\le \frac{1}{w(A_t)} \Big(\sum_{j \in [k]} \lambda_j(b_j +x_j) - \sum_{j \in [k]} \lambda_j b_j \nonumber\\
    && \hspace{3.5cm}+   \sum_{j \in [k]} \lambda_j a_j \indc{A_{t,j}\ne 0}\Big),\label{eqn_ub_extra_cons2}\\
    &&\hspace{-7mm}\le \frac{1}{w(A_t)}   \sum_{j \in [k]} \lambda_j(a_j+x_j) \indc{A_{t,j}\ne 0}).\label{eqn_ub_extra_cons3}
\end{eqnarray}In the above \eqref{eqn_ub_extra_cons2} follows from  \eqref{eqn_lmd_w_rel_sup_mod} as $b_j+x_j \in \{0,[\kappa_j]\}$, and \eqref{eqn_ub_extra_cons3} follows by the fact that  $\indc{A_{t,j}\ne 0} = 1$ when $x_j\ne 0$ or $a_j \ne 0$. Further, from \eqref{eqn_lb_mu_sup_mod1} and \eqref{eqn_ub_extra_cons3}, it is clear that \eqref{eqn_lb_mu_sup_mod1} implies \eqref{eqn_cons_At_bt_sup_mod} since $a_j+x_j \in \{0,[\kappa_j]\}$.}

\revtwo{Thus, the only binding constraints are of the form \eqref{eqn_lmd_w_rel_sup_mod} and \eqref{eqn_lb_mu_sup_mod1}. In these constraints, the coefficient of $\lambda_j$ is always positive hence at optimality,  $\lambda_j=w(n)/n$ for all $j \in [k]$. Finally, the RHS \eqref{eqn_lb_mu_sup_mod1} upon substituting $\{\lambda_j\}$ becomes
\begin{eqnarray}
    &=&\frac{w(n)}{n}\frac{1}{w\big(\sum_{j \in[k]} a_j\big)} \sum_{j \in [k]} a_j\indc{\Tilde{A}_{t,j}\ne 0}\label{eqn_mu_con_sup_mod_1}\\
    &\le& \frac{w(n)}{n}\frac{ \sum_{j \in [k]} a_j }{w\big(\sum_{j \in[k]} a_j\big)} \ \le \ \frac{w(n)}{n}. \hspace{10mm}\label{eqn_mu_con_sup_mod_2}
\end{eqnarray}In the above, \eqref{eqn_mu_con_sup_mod_2} follows from the fact that $\frac{w(j)}{j}$ is an increasing function because of supermodularity. Hence by \eqref{eq:upbd_sup_mod_poa} and \eqref{eqn_mu_con_sup_mod_2} $\mu = \frac{w(n)}{n}$ which results in the PoA of $\frac{n}{w(n)}$.} \eop

\subsubsection*{Proof of Theorem \ref{thm_set_cov_com_fail}}Consider a network $(\C,\N)$ and any utility generating mechanism  $f$. If $f_j(1) \le 0$ for any $j$, by Theorem \ref{thm_dual_lp_gen} $\poa=0$ and we have nothing to prove. Now consider any $f$ such that $f_j(1)>0$ for all $j\in [k]$. From Theorem \ref{thm_dual_lp_gen}, the constraints in the LP \eqref{eqn_poa_dual_lp_gen} deriving the  PoA when $w(\cdot)$ satisfies \eqref{eqn_def_set_cov} equals
\begin{equation}\label{eqn_cons_dual_sc}
     \mu  \indc{A_t \ne 0} \ge  \indc{B_t \ne 0} + \sum_{j \in [k]} \lambda_j  [ a_j f_j(A_{t,j})-  b_j f_j(A_{t,j}+1)]. 
\end{equation}When $t$ is such that $b_j=1$ and rest of the components of $t$ are zero, \eqref{eqn_cons_dual_sc} becomes
\begin{eqnarray}\label{eqn_lb_lmd_sc}
    \lambda_j & \ge& \frac{1}{f_j(1)} \mbox{ for any }j\in[k].
\end{eqnarray}Choose another $t\in \I_\R$ such that $a_l = 1$, $x_p=1$ and rest of the components of $t$ are zero.  Then \eqref{eqn_cons_dual_sc} becomes
\begin{eqnarray}
    \mu  &\ge&  1 +  \lambda_l  f_l(1)\ \ge\ 2
\end{eqnarray}where the second inequality follows by \eqref{eqn_lb_lmd_sc}. Thus any feasible $\mu$ for LP \eqref{eqn_poa_dual_lp_gen} is at least 2 implying the optimal $\mu^* \ge 2$. Then,
\revtwo{\begin{eqnarray}\label{eqn_proof_thm_two1}
    \poa & = & \frac{1}{\mu^*} \ \le \ \frac{1}{2} \mbox{ for all } f.
\end{eqnarray}} \eop

\begin{lemma}\label{lem_opt_lmd_sc_mc}
 Let $(\mu^*,\lambda^*_1,\dots,\lambda^*_k)$ be the solution of  LP \eqref{eqn_poa_dual_lp_gen} deriving $\poa$. If $w(\cdot)$ satisfies \eqref{eqn_def_set_cov} and   $f_j = \fmc_j$ as in \eqref{eqn_def_MC} for all $j \in [k]$, then $\lambda^*_j=1$ for all $j\in[k]$.
\end{lemma}
\textbf{Proof:} Consider the constraint in the LP \eqref{eqn_poa_dual_lp_gen} deriving the  $\poa$ when $w(\cdot)$ satisfies \eqref{eqn_def_set_cov}  given in \eqref{eqn_cons_dual_sc} at marginal contribution utility \eqref{eqn_def_MC}:
\begin{equation}\label{eqn_cons_sc_mc}
       \mu  \indc{A_t \ne 0} \ge  \indc{B_t \ne 0} + \sum_{j \in [k]} \lambda_j  [ a_j \indc{A_{t,j}=1}-  b_j \indc{A_{t,j}=0}].
\end{equation}

For any $t$ such that $b_j=1$ for some $j\in[k]$ and rest of the components are zero, the above constraint implies
 $\lambda_j \ge 1$. This holds for all $j \in [k]$.

To prove that $\lambda_j^*=1$, we show that the only binding constraints on $\mu$ have a positive coefficient of $\lambda_j$ for all $j$. Since the LP is a minimization problem,  $\lambda_j$ must take smallest possible value at the optimal, and hence must be 1 for all $j\in[k]$.

If the number of classes $k=1$, then $A_t = A_{t,1}$ and hence any constraint involving $\mu$ (i.e., with $A_t \ne 0$) has the coefficient of $\lambda_1$ as non-negative. Thus, at optimality $\lambda_1$ takes the smallest possible value and equals $\lambda^*_1=1$.

If $k\ge 2$, then any constraint involving $\mu$ (i.e., with $A_t \ne 0$) that has negative coefficient of $\lambda_j$ for any $j\in[k]$ is implied by another constraint with positive coefficient of $\lambda_j$. To see this, fix $t$ such that $A_t \ne 0$ and say $\lambda_p$ has a negative coefficient; this implies $A_{t,p}=0$ and $b_p >0$. Note that $A_t \ne 0$ implies that $a_l+x_l>0$ for some $l\in [k]$ and $l\ne p$. Then \eqref{eqn_cons_sc_mc} at such $t$ equals
\begin{eqnarray}\label{eqn_lmd_pos1}
      \mu &\ge&1 + \lambda_l a_l\indc{A_{t,l}=1} -\lambda_p b_p\nonumber\\
      && + \sum_{j \in [k]\backslash \{p,l\}} \lambda_j  [ a_j \indc{A_{t,j}=1}-  b_j \indc{A_{t,j}=0}].
\end{eqnarray} 

Consider another $\Tilde{t} \in \I_\R$ such that $b_l=1$, $b_p=0$ and rest of the components of $\Tilde{t} $ match with those at $t$. Then \eqref{eqn_cons_sc_mc}  at $\Tilde{t}$ equals
\begin{eqnarray}\label{eqn_lmd_pos2}
   \mu &\ge&1 + \lambda_l a_l\indc{A_{t,l}=1} \nonumber\\
      && + \sum_{j \in [k]\backslash \{p,l\}} \lambda_j  [ a_j \indc{A_{t,j}=1}-  b_j \indc{A_{t,j}=0}].
\end{eqnarray}Clearly, \eqref{eqn_lmd_pos2} implies \eqref{eqn_lmd_pos1}. Similar reasoning follows for all $\lambda_j$ with negative coefficients. Hence, at optimality $\lambda^*_j=1$ for all $j\in [k]$. \eop

\subsubsection*{Proof of Theorem \ref{thm_bl_is_sc}}We first prove part (i). Consider the case when $\kappa<n$. 
Then there are $\kappa+1$ classes: $\kappa$ classes for blind and isolated agents and one class for the normal agents. Without loss of generality, let $t_j$ for $j=1,\dots,\kappa_1$ correspond to blind agents, and  for $j=\kappa_1+1,\dots,\kappa$  correspond to isolated agents. Let  $t_{\kappa+1}$ correspond to normal agents.

\revtwo{Any $f$ with $f_j(1) \le 0$ for any $j\in [k]$ has the PoA zero. Now consider $f$ with $f_j(1) >0$ for all $j$}. The constraint in the LP deriving the PoA for such $f$ is given by \eqref{eqn_cons_dual_sc}. Consider a $t\in\I_\R$ with $a_1,\dots a_\kappa=1$, $x_{\kappa+1}=1$ and rest of the components zero. Then $A_t \ne 0$  and $B_t \ne 0$ so that~\eqref{eqn_cons_dual_sc} becomes
\begin{eqnarray}\label{eqn_proof_thm_three1}
    \mu& \ge & 1 + \sum_{j\in [\kappa]}\lambda_j f_j(1)   \  \ge\ 1+\kappa
\end{eqnarray}
where the second inequality follows from \eqref{eqn_lb_lmd_sc}.  Using similar logic as in \eqref{eqn_proof_thm_two1}, we obtain
\revtwo{\begin{eqnarray}\label{eqn_proof_thm_three2}
    \poa &\le& \frac{1}{1+\kappa} \mbox{ when } \kappa<n.
\end{eqnarray}}

When $\kappa=n$, then there are $n$ classes with each class having exactly one agent who cannot observe any other agent. Consider a $t\in \I_\R$ such that $a_1,\dots a_n=1$ and the rest of the components are zero. Using similar arguments as above we obtain
\revtwo{\begin{eqnarray}\label{eqn_proof_thm_three3}
    \poa &\le& \frac{1}{n} \mbox{ when } \kappa = n.
\end{eqnarray}}
Using \eqref{eqn_proof_thm_three2} and \eqref{eqn_proof_thm_three3},
\revtwo{\begin{eqnarray}\label{eqn_proof_thm_three4}
    \poa &\le&  \max\left\{\frac{1}{1+\kappa},\frac{1}{n} \right\} \mbox{ for any } \kappa \in [n]. \hspace{8mm}
\end{eqnarray}}
\hide{The authors in \cite{Grimsman2020} prove that the upper bounds in \eqref{eqn_proof_thm_three4} is achieved at marginal contribution utility $\fmc$ defined in \eqref{eqn_def_MC}.} \revtwo{Now we prove part (ii) using our LP \eqref{eqn_poa_dual_lp_gen} in the following.}

Consider $\kappa<n$, then \eqref{eqn_cons_dual_sc} at $\fmc=(\fmc_1,\dots,\fmc_j)$ equals \eqref{eqn_cons_sc_mc}. By Lemma \ref{lem_opt_lmd_sc_mc}, at optimality in LP \eqref{eqn_poa_dual_lp_gen}, $\lambda_j^*=1$ for all $j$. Thus, the optimal $\mu^*$ satisfies (let $l:=\kappa+1$),
\begin{eqnarray}\label{eqn_proof_thm_three5}
     \mu^*  \indc{A_t \ne 0} &\ge&  \sum_{j \in [\kappa]}   [ a_j \indc{a_j+x_j=1}-  b_j \indc{a_j+x_j=0}] \\
                           && +  a_l \indc{A_{t,l}=1}-  b_l \indc{A_{t,l}=0} +\indc{B_t \ne 0}.\nonumber
\end{eqnarray}
Now observe that if $t$ is such that $a_l \indc{A_{t,l}=1}=1$, then $a_j \indc{a_j+x_j=1}=0$ for all $j\in [\kappa_1]$ by the definition of blind and normal agents and \eqref{eqn_short}. The RHS of~\eqref{eqn_proof_thm_three5} for such $t$ can at most be $1+\kappa_2$ which is obtained when $a_j=1$ for all $j=\kappa_1+1,\dots,\kappa$, $b_j>0$ for some $j\in [\kappa_1]$ and rest of the components are zero. If $t\in \I_\R$ is such that $a_l \indc{A_{t,l}=1}=0$, then the RHS of \eqref{eqn_proof_thm_three5} is at most $1+\kappa$ which is obtained when $a_j=1$ for all $j\in [\kappa]$, $b_l>0$ and rest of the components are zero. 
Since $1+\kappa> 1+\kappa_2$, the optimal $\mu^*=1+\kappa$, and PoA equals $\frac{1}{1+\kappa}$.

When $\kappa=n$, then the RHS of~\eqref{eqn_proof_thm_three5} is at most $n$ which is achieved at $a_j=1$ for all $j$, thus $\mu^*=n$, and the PoA equals $\frac{1}{n}$. 
Thus,
\begin{equation}\label{eqn_proof_thm_three6}
 {\rm PoA}(\fmc,w,\C,\N) = \max\left\{\frac{1}{1+\kappa},\frac{1}{n}\right\} \mbox{ for } \kappa \in [n].
\end{equation}\revtwo{This proves part (ii).}

\revtwo{Towards part (iii), observe that the case of $\kappa \ge n-1$ is exactly the same as the case of $\kappa=n$ in part (ii), and has the  PoA of $\frac{1}{n}$ achieved at marginal contribution. Now consider and the marginal contribution utility generating mechanism and note that the optimal $\mu^*$ satisfies the following:} 
\begin{equation}
     \mu^*  \indc{A_t \ne 0}\ge \indc{B_t\ne 0}+  \sum_{j \in [\kappa+1]}   [ a_j \indc{a_j+x_j=1}-  b_j \indc{a_j+x_j=0}]
\end{equation}The most binding constraint is obtained at $t\in \I_\R$ such that $a_j=1$ for all $j\in[\kappa+1]$, $b_{j}=1$ for some $j$, and rest of the components are zero. This implies $\mu^* = 2+\kappa$, \revtwo{hence the PoA at marginal contribution equals $\frac{1}{2+\kappa}$}. Thus the proof. \eop
\subsubsection*{Proof of Theorem \ref{thm_clique}} Since there exist classes $\C_l \in \C$ and $\C_p \in \C$ such that $\C_p \not \in \infoc{l}$, by Theorem \ref{thm_set_cov_com_fail}, \revtwo{$\poa\le\frac{1}{2}$}.  
The constraint in the LP \eqref{eqn_poa_dual_lp_gen} deriving the PoA when $w(\cdot)$ satisfies \eqref{eqn_def_set_cov}  and   $f_j = \fmc_j$ as in \eqref{eqn_def_MC} for all $j \in [k]$ is given by \eqref{eqn_cons_sc_mc}. Let $(\mu^*,\lambda^*_1,\dots,\lambda^*_k)$ represent the optimal solution of this LP, then by Lemma \ref{lem_opt_lmd_sc_mc}, $\mu^*$ satisfies the following for all $t\in \I_\R$ such that $A_t\ne 0$:
\begin{equation}\label{eqn_thm4_1}
  \mu^* \ge  \indc{B_t \ne 0} + \sum_{j \in [k]}  [ a_j \indc{A_{t,j}=1}-  b_j \indc{A_{t,j}=0}].
\end{equation}
For a network $(\C,\N)$ as in theorem hypothesis, there can be at most one index $j \in [k]$ such that $ a_j \indc{A_{t,j}=1}=1$ for any $t\in \I_\R$ . To see this, say $a_1 \indc{A_{t,1}=1}$, this is possible only when $a_1=1$, $x_1=0$ and $a_j=x_j=0$ for all $j \in \O_1\backslash\{1\}$ which implies $a_j \indc{A_{t,j}=1}=0$ for such $j$. Further, all $j\in [k]\backslash\O_1$ must satisfy $\C_1\in \N_j$ by the theorem hypothesis, which  implies $a_j\indc{A_{t,j}=1}=0$ as $a_1=1$. Thus  for all $t\in \I_\R$
\begin{eqnarray*}
     \indc{B_t \ne 0} + \sum_{j \in [k]}  [ a_j \indc{A_{t,j}=1}-  b_j \indc{A_{t,j}=0}] &\le&2
\end{eqnarray*}where the equality is  achieved for a $t\in \I_\R$ such that $a_1=1$ and $b_j=1$ for some $j \ne 1$. Hence $\mu^*=2$, and the PoA at $\fmc$ equals $\frac{1}{2}$. Thus the proof. \eop


\subsubsection*{Proof of Theorem \ref{thm_gen_sc}} \revtwo{Let $\Bar{\C} \subseteq \C$ denote the set of ‘blind’ or ‘isolated’ classes, then $k_1 = |\Bar{\C}|$. Towards part (i),} the constraint in the   LP deriving PoA  for any $f$ equals \eqref{eqn_cons_dual_sc}. \revtwo{When $k_1<k$, consider} a 
$t\in \I_\R$ with $a_j =1$ for all $j\in[k]$ such that $\C_j\in \Bar{\C}$, $x_{j}=1$ for some $j\in [k]$ such that $\C_j\in \C\backslash\Bar{\C}$, and all other components zero.  Since $\lambda_j \ge \frac{1}{f_j(1)}$ by \eqref{eqn_lb_lmd_sc}, the constraint \eqref{eqn_cons_dual_sc} at such $t$ gives
\revtwo{\begin{eqnarray}\label{eqn_proof_thm_five1}
    \mu &\ge&  1 + k_1, \mbox{ thus }\nonumber\\
    \poa &=& \frac{1}{\mu^*}\ \le\ \frac{1}{\mu}\ \le\ \frac{1}{1+k_1}.
\end{eqnarray}When $k_1=k$, consider the constraint in the  LP \eqref{eqn_cons_dual_sc} deriving PoA 
for any $f$ at $t \in \I_\R$ with $a_j =1$ for all $j\in [k]$ and all other are components zero.  Since $\lambda_j \ge \frac{1}{f_j(1)}$ by \eqref{eqn_lb_lmd_sc}, the constraint becomes $\mu \ge k$  and using similar arguments as in  \eqref{eqn_proof_thm_five1}, we get
\begin{eqnarray}
\poa &\le& \frac{1}{k}.
\end{eqnarray}This proves part (i).}

\revtwo{Toward part (ii), consider the marginal contribution utility generating mechanism \eqref{eqn_def_MC}.}
The constraint for any $t\in \I_\R$ is then given by \eqref{eqn_cons_sc_mc}.  By Lemma \ref{lem_opt_lmd_sc_mc} the optimal solution $(\mu^*,\lambda^*_1,\dots,\lambda^*_k)$ satisfies \eqref{eqn_thm4_1}.
Observe that the RHS of \eqref{eqn_thm4_1} has the following upper bound:
\begin{eqnarray}\label{eqn_lb_mc_set_c}
    \indc{B_t \ne 0} + \sum_{j \in [k]} a_j \indc{A_{t,j}=1}.
\end{eqnarray}
\revtwo{Since $k_b,k_n\ge 1$, let $\C_p$ be the blind class. By definition, there exists some $\C_l$ such that $\C_p \in \N_l$. Hence} either $a_p \indc{A_{t,p}=1}=1$ or $a_l \indc{A_{t,l}=1}=1$; both terms cannot be simultaneously positive since $\C_p \in \infoc{l}$.
Hence, the term in \eqref{eqn_lb_mc_set_c} is at most $k$, thus at optimality,
\begin{eqnarray}\label{eqn_thm_5_2}
    \mu^*  \le k \mbox{ thus } {\rm PoA}(\fmc,w,\C,\N)\ =\ \frac{1}{\mu^*} \ \ge\ \frac{1}{k}.
\end{eqnarray}\revtwo{This proves part (ii).
Toward part (iii), when $k_b=0$ and $k_i,k_n \ge 1$, note that there exists no $\C_p \in \Bar{\C}$ and $\C_l \in \C\backslash\Bar{\C}$ such that $\C_p \in \infoc{l}$,} the term in  \eqref{eqn_lb_mc_set_c} is at most $1+k$, thus 
\begin{eqnarray}\label{eqn_thm_5_3}
     \mu^*  \le k \mbox{ thus } {\rm PoA}(\fmc,w,\C,\N)\ =\ \frac{1}{\mu^*} \ \ge\ \frac{1}{1+k}.
\end{eqnarray}
\revtwo{When $k_n=0$ and hence $k_1 = k_b+k_i =k$,} consider the  marginal contribution as in \eqref{eqn_def_MC}. The constraint deriving PoA for any $t\in \I_\R$ is then given by  \eqref{eqn_cons_sc_mc}. Note that the most binding constraints are the ones with $t$ such that $a_j=1$, $x_j=0$  and $b_j=1$ for all $j\in [k]$. That is,
\begin{eqnarray*}
      \mu  &\ge &  1 + \sum_{j \in [k]} \lambda_j.
\end{eqnarray*}By Lemma \ref{lem_opt_lmd_sc_mc} we get $\mu^*=1+k$, and hence \revtwo{$\poa \ge \frac{1}{1+k}$. This proves part (iii)}. \eop

\begin{IEEEbiography}[{\includegraphics[width=1in,height=1.25in,clip,keepaspectratio]{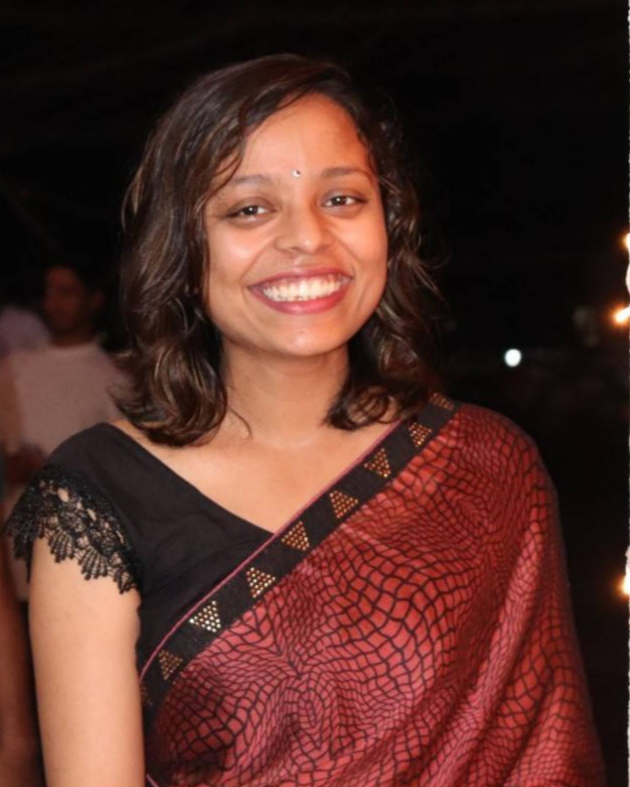}}]{Vartika Singh} (Member, IEEE) received the Bachelor of Technology in Textile Chemistry from Dr. A.P.J. Abdul Kalam Technical University, Lucknow, Uttar Pradesh, India in 2018. She received the Master of Technology and
PhD in Industrial Engineering and Operations Research from Indian Institute of Technology Bombay, Mumbai, Maharashtra, India in 2024. 

After graduating, she worked as a postdoctoral research associate in the Department of Computer Science at the University of Colorado Colorado Springs until 2025 under the supervision of Philip N. Brown.  Her research interests include multi-agent coordination using game theory and competitive Markov decision processes.

\end{IEEEbiography}

\begin{IEEEbiography}[{\includegraphics[width=1in,height=1.25in,clip,keepaspectratio]{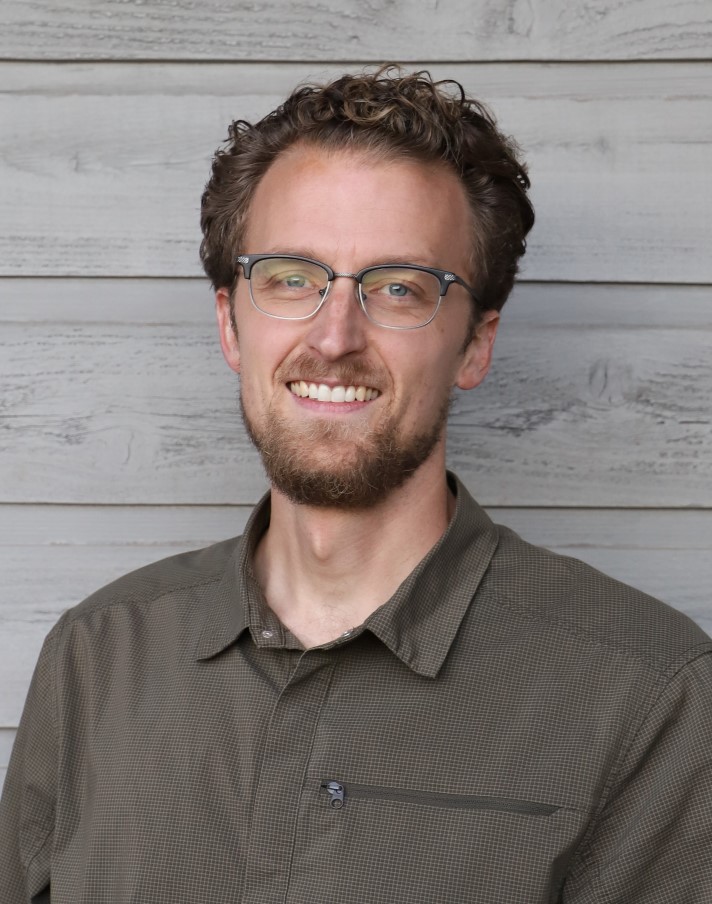}}]{Philip N. Brown} (Member, IEEE)
 is an Associate Professor in the Department of Computer Science at the University of Colorado Colorado Springs and a Spring 2026 Visiting Professor at the Politecnico di Torino. Philip received the Bachelor of Science in Electrical Engineering in 2007 from Georgia Tech. 
 He received the Master of Science in Electrical Engineering in 2015 from the University of Colorado at Boulder under the supervision of Jason R. Marden, where he was a recipient of the University of Colorado Chancellor's Fellowship. He received the PhD in Electrical and Computer Engineering from the University of California, Santa Barbara under the supervision of Jason R. Marden.
 He received the 2018 CCDC Best PhD Thesis Award from UCSB, the Best Paper Award from GameNets 2021, a 2023 AFOSR Young Investigator Program award, a 2025 NSF CAREER award, and a 2025 ARO Early Career Program award.
 Philip is interested in the interactions between engineered and social systems.
\end{IEEEbiography}

\end{document}